\documentclass[runningheads,final]{llncs}[2018/03/10]

\usepackage[T1]{fontenc}
\usepackage[utf8x]{inputenc}
\usepackage[english]{babel}

\usepackage{graphicx}
\usepackage{latexsym,stmaryrd}
\usepackage{mathpartir}
\usepackage{amsfonts,amssymb,mathtools,thmtools}

\usepackage[usenames,dvipsnames]{xcolor}
\usepackage{listings}
\usepackage{courier}

\usepackage{microtype}

\usepackage{tikz}
\usepackage{tikz-qtree}
\usetikzlibrary{shapes,arrows,positioning}

\usepackage{booktabs}
\usepackage{ebproof}
\usepackage{subcaption}
\usepackage{placeins}
\usepackage{dashbox}
\usepackage[inline]{enumitem}

\usepackage{mystyle}
\usepackage{hyperref}

\mprset{sep=1.2em} 
\def \TirNameStyle #1{\scriptsize \textsc{#1}}
\allowdisplaybreaks
\begin{document}
\title{Using Standard Typing Algorithms Incrementally}
\titlerunning{Using Standard Typing Algorithms Incrementally}

\ifdefined\ANON
    \author{\textit{Anonymous Author(s)}}
    \institute{\textit{Unknown Affiliation(s)}}
\else
    \author{
        Matteo Busi\inst{1} \and 
        Pierpaolo Degano\inst{1} \and 
        Letterio Galletta\inst{2} 
    }
    \authorrunning{Matteo Busi et al.}
    \titlerunning{Using Standard Typing Algorithms Incrementally}
    \institute{Dipartimento di Informatica, Universit\`a di Pisa, Pisa, Italy \\ \email{\{matteo.busi,degano\}@di.unipi.it}
        \and
        IMT School for Advanced Studies, Lucca, Italy \\ \email{letterio.galletta@imtlucca.it}
    }
\fi

\maketitle
\begin{abstract}
Modern languages are equipped with static type checking/inference that helps programmers to keep a clean programming style and to reduce errors.
However, the ever-growing size of programs and their continuous evolution require building fast and efficient analysers.
A promising solution is incrementality, so one only re-types those parts of the program that are new, rather than the entire codebase.
We propose an algorithmic schema driving the definition of an incremental typing algorithm that exploits the existing, standard ones with no changes.
Ours is a \emph{grey-box} approach, meaning that just the shape of the input, that of the results and some domain-specific knowledge are needed to instantiate our schema.
Here, we present the foundations of our approach and we show it at work to derive three different incremental typing algorithms.
The first two implement type checking and inference for a functional language.
The last one type-checks an imperative language to detect information flow and non-interference.
We assessed our proposal on a prototypical implementation of an incremental type checker.
Our experiments show that using the type checker incrementally is (almost) always rewarding.
\end{abstract}

\section{Introduction}\label{sec:intro}
%
Most of the modern programming languages are equipped with mechanisms
for checking or inferring types.
Such static analyses prescribe programmers a clean programming style
and help them to reduce errors.
The ever-growing size of programs requires building fast and efficient analyzers.  
This quest becomes even more demanding because many companies are recently adopting
development methodologies that advocate a continuous evolution of
software, e.g.\ \emph{perpetual development model}~\cite{CalcagnoDDGHLOP15}.
In such a model a shared code base is altered by many programmers submitting small code modifications (\emph{diffs}).
Consequently, defining static analyses and verification algorithms that require an amount of work on the size of the \emph{diffs} instead of the whole code base becomes a crucial problem, as recently observed by~\cite{facebook-harman}.

Just as software systems grow and change incrementally, also typing
should be done incrementally, so as to only re-type those parts that
are new or modified, rather than the entire codebase.
The literature reports on some techniques, briefly surveyed below,
which introduce \emph{new} typing algorithms that work incrementally.
Instead, we propose a method that \emph{makes incremental} an \emph{existing}
typing algorithm by reusing work already done, using caching and memoization.
An advantage of our proposal is that it consists of an algorithmic schema \emph{independent} of any
specific language and type system.

Roughly, our schema works as follows.
We start from the abstract syntax tree of the program, where each node 
is annotated with the result $R$ provided by the original
typing algorithm $\mathcal A$.
%
We build then a cache, containing for each subterm $t$ the
result $R$ and other relevant contextual information needed by $\mathcal A$
to type $t$ (typically a typing environment binding the free variables of $t$).
When the program changes, its annotated abstract syntax tree changes
accordingly and typing the subterm associated with the changed node
is done incrementally, by reusing the results in the cache whenever
possible and by suitably invoking $\mathcal A$ upon need.
%
Clearly, the more local the changes, the more information is reused.

Technically, our proposal consists of a set of rule schemata that
drive the usage of the cache and of the original algorithm
$\mathcal A$, as sketched above.
Actually, the user has to define the shape of caches and to
instantiate a well-confined part of the rule schemata.
If the instantiation meets an easy-to-check criterion, the typing results of $\mathcal A$ and of the incremental algorithm are guaranteed to be coherent, i.e.\ the incremental algorithm behaves as the non-incremental one.
All the above provides us with the guidelines to develop a framework that makes incremental the usage of a given typing algorithm.

Summing up, the main contributions of this paper include:
\begin{itemize}
    \item a parametric, language-independent algorithmic schema that uses an existing typing algorithm $\mathcal A$ incrementally;
    \item a formalisation of the steps that instantiate the schema and yield the incremental version of $\mathcal A$:
    %
    the resulting typing algorithm only types the \emph{diffs} and those parts of the code affected by them;
    \item a characterisation of the rule format of standard typing algorithms in terms of two functions $\tr$ and $\echeckjoin$;
    \item a theorem that under a mild condition guarantees the coherence of results between the original algorithm and its incremental version;
    \item the instantiation of the schema for two type checking and one type inference algorithm for a functional and an imperative language; 
    \item a prototype of the incremental version of the type checker for MinCaml~\cite{MinCaml}, showing that implementing the schema is doable;\footnote{Available at \url{https://github.com/mcaos/incremental-mincaml}} and
    \item 
    experimental results showing that the cost of using the type checker incrementally depends on the size of \emph{diffs}, and its performance increases as these become smaller.
\end{itemize}

\paragraph{Related work.}
To the best of our knowledge, the literature
has some proposals for incrementally typing programs.
However, these approaches heavily differ from ours, because all of them propose a \emph{new} incremental algorithm for type checking, while we are using \emph{existing} algorithms.
Additionally, none of the approaches surveyed below use a uniform characterisation of type judgements as we do through the metafunctions $\tr$ and $\echeckjoin$.


Meertens~\cite{meertens1983incremental} proposes an incremental type checking algorithm for the language \texttt{B}.
Johnson and Walz~\cite{johnson1986maximum} treat incremental type inference, focussing on identifying where type errors precisely arise.
Aditya and Nikhil~\cite{aditya1991incremental} propose an incremental Hindley/Milner type system supporting incremental type checking of top-level definitions.
Our approach instead supports incremental type-checking for all kinds of expressions, not only the top-level ones.
Miao and Siek~\cite{miao2010incremental} introduce an incremental type checker leveraging the fact that, in multi-staged programming, programs are successively refined.
Wachsmuth et al.~\cite{wachsmuth2013language} propose a task engine for type checking and name resolution: when a file is modified a task is generated and existing (cached) results are re-used where possible.
The proposal by Erdweg et al.~\cite{erdweg2015cocontextual} is the most similar to ours.
Given a type checking algorithm they describe how to obtain a \emph{new} incremental algorithm.
As in our case, they decorate an abstract syntax tree with types and typing environments, represented as sets of constraints, to be suitably propagated when typing.
In this way there is no need of dealing with top-down context propagation while types flow bottom-up.
%
%
Recently, Facebook released Pyre~\cite{facebook2018pyre} a scalable and incremental type checker for \texttt{Python}, designed to help developers of large projects.

Incrementality has also been studied for static analysis other than typing.
IncA~\cite{szabo2016inca} is a domain-specific language for the definition of incremental program analyses, which represents dependencies among the nodes of the abstract syntax tree of the target program as a graph.
Infer~\cite{facebook2018infer} uses an approach similar to ours in which analysis results are cached to improve performance~\cite{blackshear2017finding}.
Ryder and Paull~\cite{ryder1988incremental} present two incremental update algorithms, {\footnotesize ACINCB} and {\footnotesize ACINCF}, that allow incremental data-flow analysis.
Yur et al.~\cite{yur1999incremental} propose an algorithm for an incremental \emph{points-to} analysis.
McPeak et al.~\cite{mcpeak2013scalable} describe a technique for incremental and parallel static analysis based on \emph{work units} (self-contained atoms of analysis input).
The solutions are computed by a sort of processes called \emph{analysis workers}, all coordinated by an \emph{analysis master}.
Also, there are papers that use memoization with a goal similar to the one of our cache, even if they consider different
analysis techniques.
%
In particular,
Mudduluru et al.\  propose, implement, and test an incremental analysis algorithm based on memoization of (equivalent) boolean formulas used to encode paths on programs~\cite{mudduluru2014efficient}.
Some other authors also apply memoization techniques to incremental model-checking~\cite{lauterburg2008incremental,yang2009regression} and incremental symbolic execution~\cite{yang2014directed,qiu2015compositional}.

\paragraph{Plan of the paper.}
The next section intuitively presents our proposal using a
simple example.
The formalisation of our algorithmic schema for incremental typing and
its characterisation in terms of the functions $\tr$ and $\echeckjoin$
are in Section~\ref{sec:foundations}.
Sections~\ref{sec:itc-fun} and~\ref{sec:iti-fun} derive incremental
type checking and inference algorithms for a functional language,
while Section~\ref{sec:itc-imp} presents an incremental version of type checking
non-interference for an imperative language.
Section~\ref{sec:exp} briefly discusses our implementation and shows some experimental results.
The last section concludes.
All the proofs of lemmata and theorems, and more experimental results are in the Appendix.


\section{An overview of the incremental schema}\label{sec:intro-example}
%
In this section we illustrate how the algorithmic schema we propose can incrementally type check a simple program using a \emph{standard} algorithm, say $\mathcal{A}$.
Suppose to have the classical factorial program
\[
    f \triangleq \fletrec{\mathit{fact}}{\ffun{}{(n : int)}{(\fifthenelse{\fop{n}{1}{\geq}}{\fop{n}{\fapp{\mathit{fact}}{(\fop{n}{1}{-}})}{*}}{1} : int)}}{\fapp{\mathit{fact}}{x}}
\]
%
If one applies $\mathit{fact}$ to a constant greater than 0, say 7, a straightforward optimization yields the following
%
\[
    f' \triangleq \fletrec{\mathit{fact}}{\ffun{}{(n : int)}{(\fifthenelse{\fop{n}{3}{\geq}}{\fop{n}{\fapp{\mathit{fact}}{(\fop{n}{1}{-}})}{*}}{n} : int)}}{\fapp{\mathit{fact}}{7}}
\]
Suppose to have the abstract syntax tree of $f$, whose nodes are annotated with types (call it \emph{aAST}).
Now, we want to type check the new expression $f'$, by re-using as much as possible the typing information of $f$, stored in its aAST.
We proceed as follows.
First, we build a \emph{cache} $C$ associating each subexpression with its type and the typing environment needed to obtain it.
Then we \emph{incrementally} use this information to decide which existing results in the cache can be re-used and which are to be recomputed for type checking $f'$.
This process is divided into four steps.
For the moment, we omit the last one that consists in proving the correctness of the resulting algorithm.
As we will discuss later, correctness is established by showing that a component of our construction (the predicate $\ecompatenv$ used below) meets a mild condition.

\paragraph{\stepone}
The cache is a set of triples that associate with each expression $e$ the typing environment needed to close its free variables, and its type.
For example, the function application $\fapp{\mathit{fact}}{7}$, sub-expression of $f$, has the following entry in the cache, recording that $\fapp{\mathit{fact}}{7}$ has type $int$ in the typing environment
$\{\mathit{fact} \mapsto int \rightarrow int\}$:
\[
    (\fapp{\mathit{fact}}{7}, \{\mathit{fact} \mapsto int \rightarrow int\}, int)
\]

\paragraph{\steptwo}
Given an aAST for an expression $e$, we visit it in a depth-first order and we cache the relevant triples for it and for its (sub-)expressions.
Consider the sub-expression 
$e = \fifthenelse{\fop{n}{1}{\geq}}{\fop{n}{\fapp{\mathit{fact}}{(\fop{n}{1}{-}})}{*}}{1}$.
The cache records the triple $(e, \Gamma, int)$, where $e$ has type $int$ and $\Gamma =\{\mathit{fact} \mapsto int \rightarrow int, n \mapsto int\}$ gives types to the free variables of $e$.
The entries for the sub-expressions of $e$ are in \tablename~\ref{tab:cache-fun-fact} that shows the whole cache for $f$.

\begin{table}[bt]
  \caption{Tabular representation of the cache $C$ for
    $f$.}\label{tab:cache-fun-fact}
  \begin{center}
  \resizebox{1\linewidth}{!}{
            \begin{tabular}{ccc}
            \toprule
            \textbf{Expression} & \textbf{Environment} & \textbf{Type} \\
            \midrule
            $f , \quad 1, \quad 7$ & $\emptyset$ & $int$\\
            $\ffun{}{(n : int)}{(\fifthenelse{\fop{n}{1}{\geq}}{\fop{n}{\fapp{\mathit{fact}}{(\fop{n}{1}{-}})}{*}}{n} : int)}$ & $\emptyset$ & $int \rightarrow int$\\
            $n, \quad \fop{n}{1}{-}$ & $[n \mapsto int]$ & $int$\\
            $\fop{n}{1}{\geq}$ & $[n \mapsto int]$ & $bool$\\
            $\fapp{\mathit{fact}}{7}$ & $[\mathit{fact} \mapsto int \rightarrow int]$ & $int$\\
            $\mathit{fact}$ & $[\mathit{fact} \mapsto int \rightarrow int]$ & $int \rightarrow int$\\
            $\fifthenelse{\fop{n}{1}{\geq}}{\fop{n}{\fapp{\mathit{fact}}{(\fop{n}{1}{-}})}{*}}{n},
            \quad \fop{n}{\fapp{\mathit{fact}}{(\fop{n}{1}{-}})}{*}, \quad \fapp{\mathit{fact}}{(\fop{n}{1}{-})} \mbox{\ \ }$
                        & $[\mathit{fact} \mapsto int \rightarrow int, n \mapsto int]$ & $int$\\

           \bottomrule
            \end{tabular}
          }
     \end{center}
\end{table}

\paragraph{\stepthree}
A given typing algorithm $\mathcal{A}$ is used to build the incremental algorithm $\mathcal{IA}$ by following the specification given by the judgement below.
A judgement inputs an environment $\Gamma$, a cache $C$ and an expression $e$ and it computes incrementally the type $\tau$ and $C'$, with possibly updated cache entries for the sub-expressions of $e$:
%
\[
    \Gamma, C \vdash_\mathcal{IA} e : \tau \triangleright C'
\]
%
%
%
%
The incremental algorithm is expressed as a set of inductively defined rules.
Most of these simply mimic the structure of the rules defining $\mathcal{A}$.
Those for the expressions that introduce binders require instead a specific treatment of the environment and the cache.
Consider the two rules below for functional abstraction.
The first rule says that we can reuse the information available if the abstraction is cached and the environments $\Gamma$ and $\Gamma'$ coincide on the free variables of $e$ (checked by the predicate $\compatenv{\Gamma}{\Gamma'}{e}$)
\begin{align*}
\inferrule
{
    C(\ffun{}{(x : \tau_x)}{(e : \tau_e)}) = \langle \Gamma', \tau \rangle \ \ \compatenv{\Gamma}{\Gamma'}{\ffun{}{(x : \tau_x)}{(e : \tau_e)}}
}
{
    \Gamma, C \vdash_\mathcal{IA} \ffun{}{(x : \tau_x)}{(e : \tau_e)}  : \tau \triangleright C
}
\end{align*}
The second rule is for when nothing is cached (the side condition $\mathit{miss}$ holds), or the typing environments are not compatible.
In this case, $\Gamma$ is extended with the type of the argument $x$ to re-type $e$, and obtain $C''$, the update of $C$.
\begin{align*}
\inferrule*[right={$\mathit{miss}(C, e, \Gamma)$}]
{
    \Gamma[x \mapsto \tau_x], C \vdash_\mathcal{IA} e  : \tau_e \triangleright C'' \\
    C' = C'' \cup \{ (\ffun{}{(x : \tau_x)}{(e : \tau_e)}, \restrict{\Gamma}{FV(\ffun{}{(x : \tau_x)}{(e : \tau_e)}}, \tau_x \rightarrow \tau_e) \}
}
{
    \Gamma, C \vdash_\mathcal{IA} \ffun{}{(x : \tau_x)}{(e : \tau_e)}  : \tau_x \rightarrow \tau_e \triangleright C'
}
\end{align*}
%
%
Back to the example, the ``incremental'' deduction in~\figurename~\ref{fig:modified-fun-fact} suffices to type $f'$.
Note that one avoids re-checking the types of some sub-terms, e.g.\ of $\fop{n}{\fapp{\mathit{fact}}{(\fop{n}{1}{-}})}{*}$ in the proof tree of \figurename~\ref{prooftree:b}.
%
%
\begin{figure}[!bt]
    \scriptsize
    \centering
    \begin{subfigure}[b]{\textwidth}
    \centering
        \begin{prooftree}
            \hypo{C(n) = \langle \{n \mapsto int\}, int \rangle} \infer1{\Gamma, C \vdash_\mathcal{IA} n : int \triangleright C}
            \hypo{ \Gamma \vdash 3 : int} \infer1{\Gamma, C \vdash_\mathcal{IA} 3 : int \triangleright C' = C \cup \{ (3, \emptyset, int) \}}
            \infer2{\Gamma, C \vdash_\mathcal{IA} \fop{n}{3}{\geq} : bool \triangleright C'' = C' \cup \{ (\fop{n}{3}{\geq}, \{n \mapsto int\}, bool) \}}
        \end{prooftree}
    \caption{The proof tree $A$. where $\Gamma = \{n \mapsto int, \mathit{fact} \mapsto int \rightarrow int\}$}
    \end{subfigure}
    \hfill
    \bigskip

    \begin{subfigure}[b]{\textwidth}
        \begin{prooftree}
        \hypo{A}
        \hypo{C(\fop{n}{\fapp{\mathit{fact}}{(\fop{n}{1}{-}})}{*}) = \langle \{n \mapsto int, \mathit{fact} \mapsto int \rightarrow int\}, int\rangle}
        \infer1{\{n \mapsto int, \mathit{fact} \mapsto int \rightarrow int\}, C \vdash_\mathcal{IA} \fop{n}{\fapp{\mathit{fact}}{(\fop{n}{1}{-}})}{*} \triangleright C }
        \hypo{C(n) = \langle \{n \mapsto int\}, int\rangle} \infer1{\{n \mapsto int\}, C \vdash_\mathcal{IA} n \triangleright C }
        \infer3{\{n \mapsto int, \mathit{fact} \mapsto int \rightarrow int\}, C \vdash_\mathcal{IA} e_\mathit{if} : int \triangleright C''' = C'' \cup \{ (e_\mathit{if}, \{n \mapsto int, \mathit{fact} \mapsto int \rightarrow int\}, int) \}}
    \end{prooftree}
    \caption{The proof tree $B$.}\label{prooftree:b}
    \end{subfigure}
    \hfill
    \bigskip

    \begin{subfigure}[b]{\textwidth}
        \begin{prooftree}
        \hypo{B}
        \hypo{C^{iv} = C''' \cup \{ (\ffun{}{(n : int)}{(e_\mathit{if} : \tau_\mathit{if})}, \emptyset, int) \}}
        \infer2{\{\mathit{fact} \mapsto int \rightarrow int\}, C \vdash_\mathcal{IA} \ffun{}{(n : int)}{(e_\mathit{if} : \tau_\mathit{if})} : int \rightarrow int \triangleright C^{iv}}
        \hypo{C(\fapp{\mathit{fact}}{7}) = \langle \{\mathit{fact} \mapsto int \rightarrow int\}, int\rangle} \infer1{\{\mathit{fact} \mapsto int \rightarrow int\}, C \vdash_\mathcal{IA} \fapp{\mathit{fact}}{7} : int \triangleright C}
        \infer2{\emptyset, C \vdash_\mathcal{IA} f' : int \,\triangleright C^{iv} \cup \{ (f', \emptyset, int)\}}
    \end{prooftree}
    \caption{The proof tree for $\mathit{f'}$.}
    \end{subfigure}
\caption{Incremental typing of $f'$, where
\mbox{$e_\mathit{if} \triangleq \fifthenelse{\fop{n}{3}{\geq}}{\fop{n}{\fapp{\mathit{fact}}{(\fop{n}{1}{-}})}{*}}{n}$.}}\label{fig:modified-fun-fact}
\end{figure}
%



\section{Formalizing the incremental schema}\label{sec:foundations}
%
Here we formalise our algorithmic schema for incremental typing, exemplified in Section~\ref{sec:intro-example}.
Remarkably, it is independent of both the specific type system and the programming language (for that we use below $t \in Term$ to denote an expression or a statement).

Assume variables $x, y, \ldots \in Var$, types $\tau, \tau', \ldots \in Type$, and typing environments $\Gamma \colon Var \rightarrow Type \in Env$.
Also, assume that the original typing algorithm $\mathcal{A}$ is syntax-directed; that it is invoked by writing $\Gamma \vdash_\mathcal{A} t \colon R$, where $R \in Res$ is the result (not necessarily a type only); and that it is defined through inference rules.

Below we express the rules of $\mathcal{A}$ according to the following format.
It is convenient to order the subterms of $t$, by stipulating $i \leq j$ provided that $t_j$ requires the result of $t_i$ to be typed ($i,j \leq n_t$).
\begin{mathpar}
\inferrule{
    \forall i \in \mathbb{I}_t \,.\, \tr^{t}_{t_i}(\Gamma, \{R_j \}_{j < i \,\wedge\, j \in \mathbb{I}_t}) \vdash_\mathcal{A} t_i : R_i \\
    \echeckjoin_t(\Gamma, \{ R_i \}_{i \in \mathbb{I}_t}, \mathtt{out}\,R)\\
}{\Gamma \vdash_\mathcal{A} t : R}
\end{mathpar}
where $\mathbb{I}_t \subseteq \{ 1, \ldots, n_t\}$.
The function $\tr^{t}_{t_i}$ maps $\Gamma$ and a set of typing results into the typing environment needed by $t_i$.
The (conjunction of) predicate(s) $\echeckjoin_t$ checks that the subterms have compatible results $R_i$ and combines them in the overall result $R$.
(Both $\tr$ and $\echeckjoin$ are easily defined when typing rules in the usual format are rendered in the format above.)

For example the standard typing rule for variables:\footnote{
%
%
\hspace{-1mm}Instead with the axiom $\Gamma'[x \mapsto \tau] \vdash_\mathcal{A} x : \tau$ one has
$\mathbb{I}_x = \emptyset$ and the same $\echeckjoin_x$, where $\Gamma = \Gamma'[x \mapsto \tau]$.
}
%
\[
    \inferrule
    {
        x \in dom(\Gamma) \\
        \tau = \Gamma (x)
    }
    {
        \Gamma \vdash_\mathcal{A} x : \tau
      }
\]
%
is rendered in our format as follows (note that $\mathbb{I}_x = \emptyset$ just as the function $\tr$)
%
\[
\inferrule
{
    \echeckjoin_x (\Gamma, \emptyset, \mathtt{out}\,\tau)
}
{
    \Gamma \vdash_\mathcal{A} x : \tau
}
\quad \text{where} \ 
\echeckjoin_x (\Gamma, \emptyset, \mathtt{out}\,\tau) \triangleq x \in dom(\Gamma) \, \land \, \tau = \Gamma(x)
\]
As a further example consider the rule for the expression $\flet{x}{e_2}{e_3}$ below
\begin{mathpar}
    \inferrule
    {
        \Gamma \vdash_\mathcal{A} e_2 : \tau_2\\
        \Gamma[x \mapsto \tau_2] \vdash_\mathcal{A} e_3 : \tau_3
    }
    {
        \Gamma \vdash_\mathcal{A} \flet{x}{e_2}{e_3} : \tau_3
    }
\end{mathpar}
%
that becomes as follows (we abuse the set notation, e.g.\ omitting $\emptyset$ or $\{$ and $\}$).
%
\begin{mathpar}
    \inferrule
    {
        \tr^{\flet{x}{e_2}{e_3}}_{e_2}(\Gamma, \emptyset) \vdash_\mathcal{A} e_2 : \tau_2\\
        \tr^{\flet{x}{e_2}{e_3}}_{e_3}(\Gamma, \tau_2) \vdash_\mathcal{A} e_3 : \tau_3\\
        \echeckjoin_{\flet{x}{e_2}{e_3}} (\Gamma, \tau_2, \tau_3, \mathtt{out}\,\tau)
    }
    {
        \Gamma \vdash_\mathcal{A} \flet{x}{e_2}{e_3} : \tau
    }
\end{mathpar}
%
%
Note that the definition of function $\tr$ is immediate; that we need the type of $e_2$ for typing $e_3$; and that the second parameter of $\tr^{\flet{x}{e_2}{e_3}}_{e_2}$ is empty, because we only need the enviroment to type $e_2$.
\begin{equation}\label{def:tr}
    \tr^{\flet{x}{e_2}{e_3}}_{e_2}(\Gamma, \emptyset) \triangleq  \Gamma \qquad \qquad
    \tr^{\flet{x}{e_2}{e_3}}_{e_3}(\Gamma, \tau) \triangleq  \Gamma[x \mapsto \tau]
\end{equation}
Also the following definition is immediate
\[
    \echeckjoin_{\flet{x}{e_2}{e_3}} (\Gamma, \tau_2, \tau_3, \mathtt{out}\,\tau) \triangleq  (\tau = \tau_3)
\]
%
To enhance readability, we will hereto highlight the occurrences of \hltr{\tr^{t}_{t'}} (red in the pdf) and \hlcheckjoin{\echeckjoin_t} (blue in the pdf).

\paragraph{\stepone}
The shape of the cache is crucial for re-using incrementally portions of the available typing results.
A cache associates the input data $t$ and $\Gamma$ with the result $R$, rendered by a set of triples $(t, \Gamma, R)$, as done in Section~\ref{sec:intro-example}.
More formally, the set of caches $C$ is defined as:
\[
    Cache = \wp(Terms \times Env \times Res)
\]
We write $C(t) = \langle \Gamma, R \rangle$ if the cache has an entry for $t$, and $C(t) = \bot$ otherwise.
%

\paragraph{\steptwo}
Given a term, we assume that the nodes of its abstract syntax tree (called \emph{annotated abstract syntax tree} or \emph{aAST}) are annotated with the result of the typing for the subterm they represent (written $t : R$, possibly $t:\bot$ if $t$ does not type).
Let $\mathbb{I}_t$, $\{ t_i \}_{i \in \mathbb{I}_t}$, and $\tr^{t}_{t_i}$ be as above, and  let $\restrict{\Gamma}{\mathit{FV}(t)}$ be the restriction of $\Gamma$ to the free variables of $t$.
Then the following procedure visits the aAST in a depth-first manner and builds the cache.
\begin{align*}
    \buildcache{(t : R)}{\Gamma} \, =  \,& \{ (t, \restrict{\Gamma}{\mathit{FV}(t)}, R) \} \, \cup 
    \\ &
    \bigcup_{i \in \mathbb{I}_t} \big( \buildcache{(t_i : R_i)}{\hltr{\tr^{t}_{t_i}(\Gamma, \{ R_j \}_{j < i \,\land\, j \in \mathbb{I}_t})}} \big)
\end{align*}
%
The following theorem ensures that each entry of a cache returned by $\ebuildcache$ represents correct typing information.
\begin{restatable}[Cache correctness]{theorem}{thmcachecorrectness}\label{thm:cachecorrectness}
  Let $C$ be a cache, then
  \[
    (t, \Gamma, R) \in C  \iff \Gamma \vdash_\mathcal{A} t : R
  \]
\end{restatable}
%

\paragraph{\stepthree}
The third step consists of instantiating the rule templates that make typing incremental.
We remark that no change to
the original algorithm $\mathcal{A}$ is needed: it is used as a \emph{grey-box} --- what matters are just the shape of the original judgements, the rules and some domain-specific knowledge.
The judgements for the incremental typing algorithm $\mathcal{IA}$ have the form:
\[
\Gamma, C \vdash_\mathcal{IA} t : R \triangleright C'
\]

\noindent
We have three different rule templates defining the incremental typing algorithm.

The first template is for the case when there is a cache hit:
\begin{mathpar}
    \inferrule*
    {
        C(t) = \langle \Gamma', R \rangle \\
        \compatenv{\Gamma}{\Gamma'}{t}
    }
    {
        \Gamma, C \vdash_\mathcal{IA} t : R \triangleright C
    }
\end{mathpar}
where  $\compatenv{\Gamma}{\Gamma'}{t}$ is a predicate testing the compatibility of typing environments for the term $t$ and means that $\Gamma'$ includes the information represented by $\Gamma$ for $t$ and that they are compatible (see the example in Section~\ref{sec:intro-example}).
Note that this predicate must be defined for \emph{each} algorithm $\mathcal{A}$ and, as discussed below, it must meet a mild requirement to make the algorithm $\mathcal{IA}$ coherent with $\mathcal{A}$.

The second rule template is for when there is a cache miss and the term in hand has no subterms:
\begin{mathpar}
    \inferrule*[right={$\mathit{miss} (C, t, \Gamma)$}]
    {
        \Gamma \vdash_\mathcal{A} t : R\\
        C' = C \cup \{ (t, \restrict{\Gamma}{\mathit{FV}{(t)}}, R) \}
    }
    { \Gamma, C \vdash_\mathcal{IA} t : R \triangleright C'}
\end{mathpar}
where $\Gamma \vdash_\mathcal{A} t : R$ is the invocation to $\mathcal{A}$, and the predicate $\mathit{miss}$ is defined as
%
\[
    \mathit{miss} (C, t, \Gamma) \triangleq \nexists  \Gamma', R.\  \big( C(t) = \langle \Gamma', R \rangle \land \compatenv{\Gamma}{\Gamma'}{t} \big)
\]
Intuitively, this predicate means that either there is no association for $t$ in $C$, or if an association $(t, \Gamma', R)$ exists the typing environment $\Gamma'$ is not compatible with the current $\Gamma$.

Finally, the last template applies when there is a cache miss, but the term $t$ is inductively defined starting from its subterms.
In this case the rule invokes the incremental algorithm on the subterms, by composing the results available in the cache (if any):
\begin{mathpar}
    \inferrule*[right={$\mathit{miss} (C, t, \Gamma)$}]
    {
        \forall {i \in \mathbb{I}_t} \,.\, \hltr{tr_{t}^{t_i}(\Gamma, \{ R_j \}_{j < i \,\land\, j \in \mathbb{I}_t})}, C \vdash_\mathcal{IA} t_i : R_i \triangleright C^{i}\\
        \hlcheckjoin{\echeckjoin_t (\Gamma, \{ R_i \}_{i \in \mathbb{I}_t}, \mathtt{out}\,R)}\\
        C' = \{ (t, \restrict{\Gamma}{\mathit{FV}{(t)}}, R) \} \cup \bigcup_{i \in \mathbb{I}_t} C^{i}
    }
    { \Gamma, C \vdash_\mathcal{IA} t : R \triangleright C' }
\end{mathpar}

\paragraph{\stepfour}
The resulting algorithm $\mathcal{IA}$ preserves the correctness of the original one $\mathcal{A}$, provided that the rule templates above, and especially the predicate $\ecompatenv$ are carefully instantiated.

The following definition characterises when two environments are compatible, and it helps in proving that our incremental typing correctly implements the given non-incremental one.
\begin{definition}[Typing environment compatibility]\label{def:tec}
    A predicate $\ecompatenv$ \emph{expresses compatibility} iff
    \[
        \forall \, \Gamma, \Gamma', t \,.\,\compatenv{\Gamma}{\Gamma'}{t} \land \Gamma' \vdash_\mathcal{A} t : R \implies \Gamma \vdash_\mathcal{A} t : R
  \]
\end{definition}
\noindent
If the predicate $\ecompatenv$ expresses compatibility, then the incremental typing algorithm is concordant with the original one.
%
\begin{restatable}[Typing coherence]{theorem}{thmtypecoherence}\label{thm:typecoherence}
    If $\ecompatenv$ expresses compatibility,
 then for all terms $t$, caches $C$, typing environments $\Gamma$, and typing algorithm $\mathcal A$
    \[
        \Gamma \vdash_\mathcal{A} t : R \iff \Gamma, C \vdash_\mathcal{IA} t : R \triangleright C'.
    \]
\end{restatable}
%
Remarkably,  the above theorem suffices to establish the correctness of the incremental algorithm $\mathcal{IA}$,
provided that the original algorithm $\mathcal{A}$ is such.


\section{Incremental type checking for a functional language}\label{sec:itc-fun}
%
In this section we instantiate our schema in order to use incrementally the type checking algorithm of a simply typed functional programming language, called \texttt{FUN}.
The syntax, the types and the semantics of \texttt{FUN} are standard, see e.g.~\cite{nielson1999principles}.
We only recall some relevant aspects of its syntax below.
\begin{align*}
\mathit{Val} \ni v &\Coloneqq c \mid \ffun{f}{(x : \tau_x)}{(e : \tau_e)}
            \qquad \qquad \qquad \qquad \mathrel{\mathtt{op}} \in \{ +, *, =, \leq\} \\
Expr \ni e &\Coloneqq v \mid x \mid \fop{e_1}{e_2}{op} \mid \fapp{e_1}{e_2} \mid \fifthenelse{e_1}{e_2}{e_3} \mid \flet{x}{e_2}{e_3}  \\
Type \ni \tau, \tau_x, \tau_e &\Coloneqq \mathtt{int} \mid \mathtt{bool} \mid \tau_1 \rightarrow \tau_2
\qquad \qquad \qquad
Env \ni \Gamma \Coloneqq \emptyset \mid \Gamma[x \mapsto \tau]
\end{align*}
where in the functional abstraction $f$ denotes the name of the (possibly) recursive function we are defining, with type $\tau_x \rightarrow \tau_e$.
%

Assume as given the type checking algorithm $\mathcal{F}$, defined by judgements
\[
    \Gamma \vdash_\mathcal{F} e : \tau
\]
We build the type checking algorithm $\mathcal{IF}$ that uses $\mathcal{F}$ incrementally by following the four steps detailed in Section~\ref{sec:foundations}.

\paragraph{\stepone}
Each entry in the cache is a triple $(e, \Gamma, \tau)$, hence
\[
    C \in Cache = \wp(Expr \times Env \times Type)
\]

\paragraph{\steptwo}
We build the cache by visiting the aAST and ``reconstructing'' the typing environment.
%
The function $\ebuildcache$ is in~\figurename~\ref{fig:fun-tc-buildcache}, where for brevity we have directly used the results of $\tr$ rather than writing the needed invocations.
Indeed, $\tr$ is the identity almost everywhere, except for \textbf{let-in} (see equation~(\ref{def:tr})) and for abstraction where it is $\hltr{\tr^{\ffun{f}{(x : \tau_x)}{(e : \tau_e)}}_{e} (\Gamma, \{\tau_x, \tau_f\}) = \Gamma[x \mapsto \tau_x, f \mapsto \tau_f]}$.
\begin{figure}[tb!]
  \scriptsize
\begin{align*}
    \buildcache{(c : \tau_c)}{\Gamma} &\triangleq \{ (c, \emptyset, \tau_c) \}\\
    \buildcache{(x : \tau_x)}{\Gamma} &\triangleq \{ (x, [x \mapsto \tau_x], \tau_x) \}\\
    \buildcache{(\ffun{f}{(x : \tau_x)}{(e : \tau_e)} : \tau_f)}{\Gamma} &\triangleq \{ (\ffun{f}{x : \tau_x}{e : \tau_e}, \restrict{\Gamma}{\mathit{FV}(\ffun{f}{(x : \tau_x)}{(e : \tau_e)})}, \tau_f) \} \\
    &\qquad\cup (\buildcache{(f : \tau_f)}{\hltr{\Gamma}})\\
    &\qquad\cup (\buildcache{(x : \tau_x)}{\hltr{\Gamma}}) \\
    &\qquad\cup (\buildcache{(e : \tau_e)}{\hltr{\Gamma[x \mapsto \tau_x, f \mapsto \tau_f]}})\\
    \buildcache{(\flet{x}{e_2}{e_3} : \tau_{let})}{\Gamma} &\triangleq \{ (\flet{x}{e_2}{e_3}, \restrict{\Gamma}{\mathit{FV}(\flet{x}{e_2}{e_3})}, \tau_{let}) \}\\
    &\qquad\cup (\buildcache{(x : \tau_x)}{\hltr{\Gamma}})\\
    &\qquad\cup (\buildcache{(e_2 : \tau_2)}{\hltr{\Gamma}})\\
    &\qquad\cup (\buildcache{(e_3 : \tau_3)}{\hltr{\Gamma[x \mapsto \tau_x]}})\\
    \buildcache{(\fop{e_1}{e_2}{op} : \tau_{op})}{\Gamma} &\triangleq \{ (\fop{e_1}{e_2}{op}, \restrict{\Gamma}{\mathit{FV}(\fop{e_1}{e_2}{op})}, \tau_{op}) \}\\
    &\qquad\cup(\buildcache{(e_1 : \tau_1)}{\hltr{\Gamma}}) \\
    &\qquad\cup (\buildcache{(e_2 : \tau_2)}{\hltr{\Gamma}})\\
    \buildcache{(\fapp{e_1}{e_2} : \tau_{app})}{\Gamma} &\triangleq \{ (\fapp{e_1}{e_2}, \restrict{\Gamma}{\mathit{FV}(\fapp{e_1}{e_2})}, \tau_{app}) \} \\
    &\qquad\cup (\buildcache{(e_1 : \tau_1)}{\hltr{\Gamma}}) \\
    &\qquad\cup(\buildcache{(e_2 : \tau_2)}{\hltr{\Gamma}})
\end{align*}

\caption{Definition of $\ebuildcache$  for the \texttt{FUN} language.}\label{fig:fun-tc-buildcache}
\end{figure}

\paragraph{\stepthree}
By instantiating the patterns of Section~\ref{sec:foundations} we obtain judgements of the form
%
\[
\Gamma, C \vdash_\mathcal{IF} e : \tau \triangleright C'
\]
meaning that the expression $e$ has type $\tau$ in the environment $\Gamma$ and using the cache $C$.
The cache $C'$ records new discoveries during the incremental typing.

The incremental rules are in \figurename~\ref{fig:fun-simple-types}.
Most of them are trivial as they mimic the behaviour of the original algorithm ${\mathcal F}$.
As done in \figurename~\ref{fig:fun-tc-buildcache}, we simply write the results of $\hltr{\tr}$ and of $\hlcheckjoin{\echeckjoin}$, rather than their invocations.
%
%
Consider for example the rule~\rulename{($\mathcal{IF}$-Let-Miss)}: first, the subexpressions $e_2$ and $e_3$ are incrementally type checked in the environments prescribed by the relevant calls to the function $\tr$ in equation~(\ref{def:tr}), i.e.\
$\hltr{\Gamma}$ and $\hltr{\Gamma[x \mapsto \tau_2]}$, respectively.
Then, the type of the whole expression \textbf{let-in} is computed by $\hlcheckjoin{\echeckjoin_{\flet{x}{e_2}{e_3}}(\Gamma, \{\tau_2, \tau_3\}, {\mathtt {out}}\, \tau)}$.
\begin{figure}[!tb]
\scriptsize
\begin{mathpar}
    \inferrule*[lab={($\mathcal{IF}$-Hit)}]%
    {
        C(e) = \langle \Gamma', \tau \rangle \\
        \compatenv{\Gamma}{\Gamma'}{e}
    }
    {
        \Gamma, C \vdash_\mathcal{IF} e : \tau \triangleright C
    }

    \inferrule*[lab={($\mathcal{IF}$-Const-Miss)},right={$\mathit{miss} (C, c, \Gamma)$}]%
    {
        \Gamma \vdash_\mathcal{F} c : \tau\\
        C' = C \cup \{ (c, \emptyset, \tau) \}
    }
    {
        \Gamma, C \vdash_\mathcal{IF} c : \tau \triangleright C'
    }

    \inferrule*[lab={($\mathcal{IF}$-Var-Miss)},right={$\mathit{miss} (C, x, \Gamma)$}]%
    {
        \Gamma \vdash_\mathcal{F} x : \tau\\
        C' = C \cup \{ (x, \restrict{\Gamma}{x}, \tau) \}
    }
    {
        \Gamma, C \vdash_\mathcal{IF} x : \tau \triangleright C'
    }

    \inferrule*[lab={($\mathcal{IF}$-Abs-Miss)},right={$\mathit{miss} (C, \ffun{f}{(x : \tau_x)}{(e : \tau_e)}, \Gamma)$}]%
    {
        \hltr{\Gamma[x \mapsto \tau_x, f \mapsto \tau_x \rightarrow \tau_e]}, C \vdash_\mathcal{IF} e  : \tau_{body} \triangleright C'' \\
        \hlcheckjoin{\tau_{body} = \tau_e \land \tau = \tau_x \rightarrow \tau_e} \\
        C' = C'' \cup \{ (\ffun{f}{x : \tau_x}{e : \tau_e}, \restrict{\Gamma}{\mathit{FV}(\ffun{f}{(x : \tau_x)}{(e : \tau_e)})}, \tau) \}
    }
    {
        \Gamma, C \vdash_\mathcal{IF} \ffun{f}{(x : \tau_x)}{(e : \tau_e)}  : \tau \triangleright C'
    }

    \inferrule*[lab={($\mathcal{IF}$-Op-Miss)},right={$\mathit{miss} (C, \fop{e_1}{e_2}{op}, \Gamma)$}]%
    {
        \hltr{\Gamma}, C \vdash_\mathcal{IF} e_1 : \tau_1 \triangleright C'' \\
        \hltr{\Gamma}, C \vdash_\mathcal{IF} e_2 : \tau_2 \triangleright C''' \\
        \hlcheckjoin{\tau_1 = \tau_2 \land \tau = \tau_1} \\
        C' = C'' \cup C''' \cup \{ (\fop{e_1}{e_2}{op}, \restrict{\Gamma}{\mathit{FV}(\fop{e_1}{e_2}{op})}, \tau) \}
    }
    {
        \Gamma, C \vdash_\mathcal{IF} \fop{e_1}{e_2}{op}  : \tau \triangleright C'
    }

    \inferrule*[lab={($\mathcal{IF}$-App-Miss)},right={$\mathit{miss} (C, \fapp{e_1}{e_2}, \Gamma)$}]%
    {
        \hltr{\Gamma}, C \vdash_\mathcal{IF} e_1 : \tau_x \rightarrow \tau_e \triangleright C'' \\
        \hltr{\Gamma}, C \vdash_\mathcal{IF} e_2 : \tau_2 \triangleright C''' \\
        \hlcheckjoin{\tau_x = \tau_2 \land \tau = \tau_e} \\
        C' = C'' \cup C''' \cup \{ (\fapp{e_1}{e_2}, \restrict{\Gamma}{\mathit{FV}(\fapp{e_1}{e_2})}, \tau) \}
    }
    {
        \Gamma, C \vdash_\mathcal{IF} \fapp{e_1}{e_2} : \tau \triangleright C'
    }

    \inferrule*[lab={($\mathcal{IF}$-If-Miss)},right={$\mathit{miss} (C, \fifthenelse{e_1}{e_2}{e_3}, \Gamma)$}]%
    {
        \hltr{\Gamma}, C \vdash_\mathcal{IF} e_1 : \tau_1 \triangleright C'' \\
        \hltr{\Gamma}, C \vdash_\mathcal{IF} e_2 : \tau_2 \triangleright C''' \\
        \hltr{\Gamma}, C \vdash_\mathcal{IF} e_3 : \tau_3 \triangleright C^{iv} \\
        \hlcheckjoin{\tau_2 = \tau_3 \land \tau_1 = bool \land \tau = \tau_2} \\
        C' = C'' \cup C''' \cup C^{iv} \cup \{ (\fifthenelse{e_1}{e_2}{e_3}, \restrict{\Gamma}{\mathit{FV}(\fifthenelse{e_1}{e_2}{e_3})}, \tau) \}
    }
    {
        \Gamma, C \vdash_\mathcal{IF} \fifthenelse{e_1}{e_2}{e_3} : \tau \triangleright C'
    }

    \inferrule*[lab={($\mathcal{IF}$-Let-Miss)},right={$\mathit{miss} (C, \flet{x}{e_2}{e_3}, \Gamma)$}]%
    {
        \hltr{\Gamma}, C \vdash_\mathcal{IF} e_2 : \tau_2 \triangleright C'' \\
        \hltr{\Gamma[x \mapsto \tau_2]}, C \vdash_\mathcal{IF} e_3 : \tau_3 \triangleright C''' \\
        \hlcheckjoin{\tau = \tau_3}\\
        C' = C'' \cup C''' \cup \{ (\flet{x}{e_2}{e_3}, \restrict{\Gamma}{\mathit{FV}(\flet{x}{e_2}{e_3})}, \tau) \}
    }
    {
        \Gamma, C \vdash_\mathcal{IF} \flet{x}{e_2}{e_3} : \tau \triangleright C'
    }
  \end{mathpar}
\begin{equation}\label{def:comp-fun-tc}
    \text{with}\quad \compatenv{\Gamma}{\Gamma'}{e} \triangleq dom(\Gamma) \supseteq \mathit{FV}(e) \land dom(\Gamma') \supseteq \mathit{FV}(e) \land \forall y \in \mathit{FV}(e) \,.\, \Gamma(y) = \Gamma'(y)
  \end{equation}
  \caption{Rules defining incremental algorithm $\mathcal{IF}$ to type check \texttt{FUN}.}\label{fig:fun-simple-types}
  \vspace{-3mm}
\end{figure}

\paragraph{\stepfour}
To prove that $\mathcal{IF}$ is coherent with $\mathcal{F}$, we first show that $\ecompatenv$ satisfies Definition~\ref{def:tec}.
\begin{restatable}{lemma}{lemmaIFcoherence}\label{lemma:IF-coherence}
The predicate
    $\ecompatenv$ of Eq.~(\ref{def:comp-fun-tc}) in \figurename~\ref{fig:fun-simple-types} expresses compatibility.
\end{restatable}
%
%
The above lemma suffices to prove the following theorem, which is an instance of Theorem~\ref{thm:typecoherence}.
\begin{restatable}{theorem}{thmIFcoherence}\label{thmIFcoherence}
$ \quad
\forall \Gamma, C, e.\
\Gamma, C \vdash_\mathcal{F} e : \tau  \iff
\Gamma, C \vdash_\mathcal{IF} e : \tau \triangleright C'
$
\end{restatable}
%


\section{Incremental type inference for a functional language}\label{sec:iti-fun}
%
In this section we instantiate our schema to make incremental the type inference of \texttt{FUN}.
The syntax of the language is slightly modified to remove type annotations, while types are now augmented with type variables $\alpha, \beta, \ldots \in \mathit{TVar}$:
\begin{align*}
\mathit{Val} \ni v &\Coloneqq c \mid \ffun{f}{x}{e}
   \qquad \qquad \qquad \qquad \mathrel{\mathtt{op}} \in \{ +, *, =, \leq\} \\
\mathit{Expr} \ni e &\Coloneqq v \mid x \mid \fop{e_1}{e_2}{op} \mid \fapp{e_1}{e_2} \mid \fifthenelse{e_1}{e_2}{e_3} \mid \flet{x}{e_2}{e_3} \\
\mathit{AType} \ni \tau &\Coloneqq \mathtt{int} \mid \mathtt{bool} \mid \tau_1 \rightarrow \tau_2 \mid \alpha
\qquad \mathit{Env} \ni \Gamma \Coloneqq \emptyset \mid \Gamma[x \mapsto \tau]
\end{align*}
The judgements of the type inference algorithm $\mathcal{W}$ have the form
\[
    \Gamma \vdash_\mathcal{W} e \colon (\tau, \theta)
\]
where $\theta \colon (\mathit{TVar} \rightarrow \mathit{AType}) \in \mathit{Subst}$ is a substitution mapping type variables into augmented types.
As usual, we write $\theta\,\tau$ to indicate the application of the substitution $\theta$ to $\tau$, and  $\theta_2 \circ \theta_1$ stands for the composition of substitutions.

In~\figurename~\ref{fig:algo-w} we restate the inference algorithm $\mathcal W$ (see e.g.~\cite{nielson1999principles}), where we assume constants $c$ to have a fixed and known type, and $\mathcal{U}$ to be the standard type unification algorithm.
In the resulting set of rules we have coloured and framed the parts that drive the definitions
of $\hltr{\tr}$ and $\hlcheckjoin{\echeckjoin}$, so making clear that they occur unchanged in the
incremental inference algorithm $\mathcal{IW}$.
%
%
%
%
 \begin{figure}[tb]
   \scriptsize
     \begin{mathpar}
         \inferrule*[lab={($\mathcal{W}$-Const)}]%
         { }
         {
             \Gamma \vdash_\mathcal{W} c : (\tau_c, \mathit{id})
         }

         \inferrule*[lab={($\mathcal{W}$-Var)}]%
         { }
         {
             \Gamma \vdash_\mathcal{W} x : (\Gamma(x), \mathit{id})
         }

         \inferrule*[lab={($\mathcal{W}$-Abs)},right={$\alpha_x, \alpha_e \,\textit{fresh}$}]%
         {
             \hltr{\Gamma[x \mapsto \alpha_x, f \mapsto \alpha_x \rightarrow \alpha_e]} \vdash_\mathcal{W} e  : (\tau_e, \theta_e) \\
             \hlcheckjoin{
                 \theta_1 = \mathcal{U} (\tau_e, \theta_e \alpha_e)
             \land
             (\tau, \theta) = \big( (\theta_1\,(\theta_e\,\alpha_x)) \rightarrow (\theta_1\,\tau_e), \theta_1 \circ \theta_e \big)}
         }
         {
             \Gamma \vdash_\mathcal{W} \ffun{f}{x}{e} : (\tau, \theta)
         }

         \inferrule*[lab={($\mathcal{W}$-Op)},right={$\tau_{op}, \tau_{res} = \{ \mathtt{int}, \mathtt{bool}\}$}]%
         {
             \hltr{\Gamma} \vdash_\mathcal{W} e_1 : (\tau_1, \theta_1)\\
             \hltr{\theta_1\,\Gamma}\vdash_\mathcal{W} e_2 : (\tau_2, \theta_2)\\
             \hlcheckjoin{\theta_3 = \mathcal{U}(\theta_2\,\tau_1, \tau_{op}) \land  \theta_4 = \mathcal{U}(\theta_3\,\tau_2, \tau_{op}) \land (\tau, \theta) = (\tau_{res}, \theta_4 \circ \theta_3 \circ \theta_2 \circ \theta_1 )}
         }
         {
             \Gamma \vdash_\mathcal{W} \fop{e_1}{e_2}{op}  : (\tau, \theta)
         }

         \inferrule*[lab={($\mathcal{W}$-App)},right={$\alpha \,\textit{ fresh}$}]%
         {
             \hltr{\Gamma} \vdash_\mathcal{W} e_1 : (\tau_1, \theta_1) \rightarrow \tau_e \\
             \hltr{\theta_1\,\Gamma}\vdash_\mathcal{W} e_2 : (\tau_2, \theta_2)\\
             \hlcheckjoin{ \theta_3 = \mathcal{U} (\theta_2\,\tau_1, \tau_2 \rightarrow \alpha) \land (\tau, \theta)= (\theta_3\,\alpha, \theta_3 \circ \theta_2 \circ \theta_1)}
         }
         {
             \Gamma\vdash_\mathcal{W} \fapp{e_1}{e_2} : (\tau, \theta)
         }

         \inferrule*[lab={($\mathcal{W}$-If)}]%
         {
             \hltr{\Gamma}\vdash_\mathcal{W} e_1 : (\tau_1, \theta_1)\\
             \hltr{\theta_1\,\Gamma}\vdash_\mathcal{W} e_2 : (\tau_2, \theta_2)\\
             \hltr{\theta_2 (\theta_1\,\Gamma)}\vdash_\mathcal{W} e_3 : (\tau_3, \theta_3)\\
             \hlcheckjoin{\theta_4 = \mathcal{U}(\theta_3 (\theta_2\,\tau_1), bool) \land \theta_5 = \mathcal{U}(\theta_4\,\tau_3, \theta_4 (\theta_3\,\tau_1)) \land
             (\tau, \theta) = (\theta_5 (\theta_4\,\tau_3), \theta_5 \circ \theta_4 \circ \theta_3 \circ \theta_2)}
         }
         {
             \Gamma \vdash_\mathcal{W} \fifthenelse{e_1}{e_2}{e_3} : (\tau, \theta)
         }

         \inferrule*[lab={($\mathcal{W}$-Let)}]%
         {
             \hltr{\Gamma}\vdash_\mathcal{W} e_2 : (\tau_2, \theta_2)\\
             \hltr{(\theta_2\,\Gamma)[x \mapsto \tau_2]} \vdash_\mathcal{W} e_3 : (\tau_3, \theta_3) \\
             \hlcheckjoin{(\tau, \theta) = (\tau_3, \theta_3 \circ \theta_2)}
         }
         {
             \Gamma \vdash_\mathcal{W} \flet{x}{e_2}{e_3} : (\tau, \theta)
         }
     \end{mathpar}
     \caption{Rules defining algorithm $\mathcal{W}$ to infer \texttt{FUN} types.}\label{fig:algo-w}
 \end{figure}

\paragraph{\stepone}
Each entry in the cache is a triple $(e, \Gamma, (\tau, \theta))$, so a cache is
\[
    C \in Cache = \wp(Expr \times Env \times (AType \times Subst))
\]

\paragraph{\steptwo}
%
The function $\ebuildcache$ is easily defined in~\figurename~\ref{fig:fun-ti-buildcache} with the environments resulting from $\hltr{\tr}$ above.
\begin{figure}[htb!]
  \scriptsize
\begin{align*}
    \buildcache{(c : (\tau_c, \theta))}{\Gamma} &\triangleq \{ (c, \emptyset, (\tau_c, \theta)) \}\\
    \buildcache{(x : (\tau_x, \theta))}{\Gamma} &\triangleq \{ (x, [x \mapsto \tau_x], (\tau_x, \theta)) \}\\
    \buildcache{(\ffun{f}{x}{e} : (\tau_f, \theta_f))}{\Gamma} &\triangleq \{ (\ffun{f}{x}{e}, \restrict{\Gamma}{\mathit{FV}(\ffun{f}{x}{e})}, (\tau_f, \theta_f)) \} \\
    &\qquad\cup (\buildcache{(f : (\tau_f, \theta_f))}{\hltr{\Gamma}})\\
    &\qquad\cup (\buildcache{(x : (\tau_x, \theta_x))}{\hltr{\Gamma}}) \\
    &\qquad\cup (\buildcache{(e : (\tau_e, \theta_e))}{\hltr{\Gamma[x \mapsto \tau_x, f \mapsto \tau_f]}})\\
    \buildcache{(\flet{x}{e_2}{e_3} : (\tau_{let}, \theta_{let}))}{\Gamma} &\triangleq \{ (\flet{x}{e_2}{e_3}, \restrict{\Gamma}{\mathit{FV}(\flet{x}{e_2}{e_3})}, (\tau_{let}, \theta_{let})) \}\\
    &\qquad\cup (\buildcache{(x : (\tau_x, \theta_x))}{\hltr{\Gamma}})\\
    &\qquad\cup (\buildcache{(e_2 : (\tau_2, \theta_2))}{\hltr{\Gamma}})\\
    &\qquad\cup (\buildcache{(e_3 : (\tau_3, \theta_3))}{\hltr{\Gamma[x \mapsto \tau_x]}})\\
    \buildcache{(\fop{e_1}{e_2}{op} : (\tau_{op}, \theta_{op}))}{\Gamma} &\triangleq \{ (\fop{e_1}{e_2}{op}, \restrict{\Gamma}{\mathit{FV}(\fop{e_1}{e_2}{op})}, (\tau_{op}, \theta_{op})) \}\\
    &\qquad\cup(\buildcache{(e_1 : (\tau_1, \theta_1)}{\hltr{\Gamma}}) \\
    &\qquad\cup (\buildcache{(e_2 : (\tau_2, \theta_2)}{\hltr{\Gamma}})\\
    \buildcache{(\fapp{e_1}{e_2} : (\tau_{app}, \theta_{app}))}{\Gamma} &\triangleq \{ (\fapp{e_1}{e_2}, \restrict{\Gamma}{\mathit{FV}(\fapp{e_1}{e_2})}, (\tau_{app}, \theta_{app})) \} \\
    &\qquad\cup (\buildcache{(e_1 : (\tau_1, \theta_1))}{\hltr{\Gamma}}) \\
    &\qquad\cup(\buildcache{(e_2 : (\tau_2, \theta_2))}{\hltr{\Gamma}})
\end{align*}

\caption{Definition of $\ebuildcache$ for the incremental type inference of \texttt{FUN}.}\label{fig:fun-ti-buildcache}
\end{figure}

\paragraph{\stepthree}
In \figurename~\ref{fig:fun-simple-inference} we display the rules defining the algorithm $\mathcal{IW}$ with judgements of the following form
\[
\Gamma, C \vdash_\mathcal{IW} e : (\tau, \theta) \triangleright C'
\]
\noindent
Most of the rules
mimic the behaviour of algorithm $\mathcal{W}$, following the templates of Section~\ref{sec:foundations}.
Consider for example the rule~\rulename{($\mathcal{IW}$-Let-Miss)}: first, the types of $e_1$ and $e_2$ are incrementally inferred in the environments prescribed by the relevant calls to the function $\tr$.
The result associated with the whole expression \textbf{let-in} is then the pair $(\tau_2, \theta_2 \circ \theta_1)$, where $\theta_1$ and $\theta_2$ are the substitutions obtained recursively from $e_1$ and $e_2$, respectively.

Consider the term $\hat{e} = \fop{n}{\fapp{\mathit{fact}}{(\fop{n}{1}{-}})}{*}$ discussed at the end of Section~\ref{sec:intro-example}.
Since in the entry for $\hat{e}$ the cache records the required substitution besides the augmented type, we save running time for the inference.
\begin{figure}[!htb]
\scriptsize
\begin{mathpar}
{   \inferrule*[lab={($\mathcal{IW}$-Hit)}]%
    {
        C(e) = \langle \Gamma', (\tau, \theta) \rangle \\
        \compatenv{\Gamma}{\Gamma'}{e}
    }
    {
        \Gamma, C \vdash_{\mathcal{IW}} e : (\tau, \theta) \triangleright C
    }

    \inferrule*[lab={($\mathcal{IW}$-Const-Miss)},right={$\mathit{miss} (C, c, \Gamma)$}]%
    {
        \Gamma \vdash_\mathcal{W} c : (\tau, \theta)\\
        C' = C \cup \{ (c, \emptyset,  (\tau, \theta)) \}
    }
    {
        \Gamma, C \vdash_{\mathcal{IW}} c : (\tau, \theta) \triangleright C'
    }
  }

    \inferrule*[lab={($\mathcal{IW}$-Var-Miss)},right={$\mathit{miss} (C, x, \Gamma)$}]%
    {
        \Gamma \vdash_\mathcal{W} x : (\tau, \theta)\\
        C' = C \cup \{ (x, \restrict{\Gamma}{x},  (\tau, \theta)) \}
    }
    {
        \Gamma, C \vdash_{\mathcal{IW}} x : (\tau, \theta) \triangleright C'
    }

    \inferrule*[lab={($\mathcal{IW}$-Abs-Miss)},right={$\mathit{miss} (C, \ffun{f}{x}{e}, \Gamma) \land \alpha_x, \alpha_e \,\textit{fresh}$}]%
    {
        \hltr{\Gamma[x \mapsto \alpha_x, f \mapsto \alpha_x \rightarrow \alpha_e]}, C \vdash_{\mathcal{IW}} e  : (\tau_e, \theta_e) \triangleright C'' \\
        \hlcheckjoin{
            \theta_1 = \mathcal{U} (\tau_e, \theta_e \alpha_e)
        \land
        (\tau, \theta) = \big( (\theta_1\,(\theta_e\,\alpha_x)) \rightarrow (\theta_1\,\tau_e), \theta_1 \circ \theta_e \big)} \\
        C' = C'' \cup \{ (\ffun{f}{x : \tau_x}{e}, \restrict{\Gamma}{\mathit{FV}(\ffun{f}{x : \tau_x}{e}}, (\tau, \theta)) \}
    }
    {
        \Gamma, C \vdash_{\mathcal{IW}} \ffun{f}{x}{e} : (\tau, \theta) \triangleright C'
    }

    \inferrule*[lab={($\mathcal{IW}$-Op-Miss)},right={$\mathit{miss} (C, \fop{e_1}{e_2}{op}, \Gamma) \land \tau_{op}, \tau_{res} = \{ \mathtt{int}, \mathtt{bool}\}$}]%
    {
        \hltr{\Gamma}, C \vdash_{\mathcal{IW}} e_1 : (\tau_1, \theta_1) \triangleright C'' \\
        \hltr{\theta_1\,\Gamma}, C \vdash_{\mathcal{IW}} e_2 : (\tau_2, \theta_2) \triangleright C''' \\
        \hlcheckjoin{\theta_3 = \mathcal{U}(\theta_2\,\tau_1, \tau_{op}) \land  \theta_4 = \mathcal{U}(\theta_3\,\tau_2, \tau_{op}) \land (\tau, \theta) = (\tau_{res}, \theta_4 \circ \theta_3 \circ \theta_2 \circ \theta_1 )} \\
        C' = C'' \cup C''' \cup \{ (\fop{e_1}{e_2}{op}, \restrict{\Gamma}{\mathit{FV}(\fop{e_1}{e_2}{op})}, (\tau, \theta)) \}
    }
    {
        \Gamma, C \vdash_{\mathcal{IW}} \fop{e_1}{e_2}{op}  : (\tau, \theta) \triangleright C'
    }

    \inferrule*[lab={($\mathcal{IW}$-App-Miss)},right={$\mathit{miss} (C, \fapp{e_1}{e_2}, \Gamma) \land \alpha \,\textit{ fresh}$}]%
    {
        \hltr{\Gamma}, C \vdash_{\mathcal{IW}} e_1 : (\tau_1, \theta_1) \rightarrow \tau_e \triangleright C'' \\
        \hltr{\theta_1\,\Gamma}, C \vdash_{\mathcal{IW}} e_2 : (\tau_2, \theta_2) \triangleright C''' \\
        \hlcheckjoin{ \theta_3 = \mathcal{U} (\theta_2\,\tau_1, \tau_2 \rightarrow \alpha) \land (\tau, \theta)= (\theta_3\,\alpha, \theta_3 \circ \theta_2 \circ \theta_1)} \\
        C' = C'' \cup C''' \cup \{ (\fapp{e_1}{e_2}, \restrict{\Gamma}{\mathit{FV}(\fapp{e_1}{e_2})}, (\tau, \theta)) \}
    }
    {
        \Gamma, C \vdash_{\mathcal{IW}} \fapp{e_1}{e_2} : (\tau, \theta) \triangleright C'
    }

    \inferrule*[lab={($\mathcal{IW}$-If-Miss)},right={$\mathit{miss} (C, \fifthenelse{e_1}{e_2}{e_3}, \Gamma)$}]%
    {
        \hltr{\Gamma}, C \vdash_{\mathcal{IW}} e_1 : (\tau_1, \theta_1) \triangleright C'' \\
        \hltr{\theta_1\,\Gamma}, C \vdash_{\mathcal{IW}} e_2 : (\tau_2, \theta_2) \triangleright C''' \\
        \hltr{\theta_2 (\theta_1\,\Gamma)}, C \vdash_{\mathcal{IW}} e_3 : (\tau_3, \theta_3) \triangleright C^{iv} \\
        \hlcheckjoin{\theta_4 = \mathcal{U}(\theta_3 (\theta_2\,\tau_1), bool) \land \theta_5 = \mathcal{U}(\theta_4\,\tau_3, \theta_4 (\theta_3\,\tau_1)) \land
        (\tau, \theta) = (\theta_5 (\theta_4\,\tau_3), \theta_5 \circ \theta_4 \circ \theta_3 \circ \theta_2)} \\
        C' = C'' \cup C''' \cup C^{iv} \cup \{ (\fifthenelse{e_1}{e_2}{e_3}, \restrict{\Gamma}{\mathit{FV}(\fifthenelse{e_1}{e_2}{e_3})}, (\tau, \theta)) \}
    }
    {
        \Gamma, C \vdash_{\mathcal{IW}} \fifthenelse{e_1}{e_2}{e_3} : (\tau, \theta) \triangleright C'
    }

    \inferrule*[lab={($\mathcal{IW}$-Let-Miss)},right={$\mathit{miss} (C, \flet{x}{e_1}{e_3}, \Gamma)$}]%
    {
        \hltr{\Gamma}, C \vdash_{\mathcal{IW}} e_2 : (\tau_2, \theta_2) \triangleright C'' \\
        \hltr{(\theta_1\,\Gamma)[x \mapsto \tau_2]}, C \vdash_{\mathcal{IW}} e_3 : (\tau_3, \theta_3) \triangleright C''' \\
        \hlcheckjoin{(\tau, \theta) = (\tau_3, \theta_3 \circ \theta_1)}\\
        C' = C'' \cup C''' \cup \{ (\flet{x}{e_2}{e_3}, \restrict{\Gamma}{\mathit{FV}(\flet{x}{e_2}{e_3})}, (\tau, \theta)) \}
    }
    {
        \Gamma, C \vdash_{\mathcal{IW}} \flet{x}{e_2}{e_3} : (\tau, \theta) \triangleright C'
    }
\end{mathpar}
\begin{equation}\label{def:comp-fun-ti}
   \text{with}\,\, \compatenv{\Gamma}{\Gamma'}{e} \triangleq dom(\Gamma) \supseteq \mathit{FV}(e) \land dom(\Gamma') \supseteq \mathit{FV}(e) \land \forall y \in \mathit{FV}(e) \,.\, \mathcal{U} (\Gamma(y), \Gamma'(y))
\end{equation}
\caption{Rules defining incremental algorithm $\mathcal{IW}$ to infer \texttt{FUN} types.}\label{fig:fun-simple-inference}
\end{figure}

\paragraph{\stepfour}
To prove the incremental algorithm $\mathcal{IW}$ coherent with $\mathcal{W}$, we first show that $\ecompatenv$ satisfies Definition~\ref{def:tec}.
\begin{restatable}{lemma}{lemmaIWcoherence}\label{lemma:IW-coherence}
The predicate
    $\ecompatenv$ of Eq.~(\ref{def:comp-fun-ti}) in \figurename~\ref{fig:fun-simple-inference} expresses compatibility.
\end{restatable}
%
%
Again, the following theorem is an instance of Theorem~\ref{thm:typecoherence}, and follows from the above lemma.
\begin{restatable}{theorem}{thmIWcoherence}\label{thm-IW-coherence}
$ \quad
\forall \Gamma, C, e.\
\Gamma    \vdash_\mathcal{W} e : (\tau, \theta)  \iff
\Gamma, C \vdash_\mathcal{IW} e : (\tau, \theta) \triangleright C'
$
\end{restatable}


\section{Incremental checking of non-interference}\label{sec:itc-imp}
%
Here we use incrementally the typing algorithm $\mathcal S$ of Volpano-Smith-Irvine~\cite{volpano1996sound,smith2006principles} for checking non-interference policies, obtaining the algorithm $\mathcal{IS}$.
We assume that the variables of programs are classified either as high, $H$, or low $L$.
Intuitively, a program enjoys the non-interference property when the values of low level variables do not depend on those of high level ones.

As usual, assume a simple imperative language {\footnotesize \texttt{WHILE}}, whose syntax is below ($\mathit{Var}$ denotes the set of program variables).
%
\begin{align*}
AExpr \ni a &\Coloneqq n \mid x \mid \wop{a_1}{a_2}{op_a}
\qquad \qquad n \in \mathbb{N}, \quad \mathtt{op_a} \in \{\mathtt{+}, \mathtt{*}, \mathtt{-}, \ldots\}, \quad x \in\mathit{Var} \\
BExpr \ni b &\Coloneqq \mathtt{true} \mid \mathtt{false} \mid \wop{b_1}{b_2}{or} \mid \wop{}{b}{not} \mid \wop{a_1}{a_2}{\leq}\\
Stmt \ni c &\Coloneqq \wskip \mid \wassign{x}{a} \mid \wseq{c_1}{c_2} \mid \wifthenelse{b}{c_1}{c_2} \mid \wwhile{b}{c}\\
Phrase \ni p &\Coloneqq a \mid b \mid c\\
DType \ni \tau & \Coloneqq H \mid L
\quad
PType \ni \varsigma  \Coloneqq \tau \mid \tau\,\mathit{var} \mid \tau\,\mathit{cmd}
\quad
Env \ni \Gamma \Coloneqq \emptyset \mid \Gamma[p \mapsto \varsigma]
\end{align*}
The type checking algorithm has judgements of the form
\begin{align*}
    \Gamma \vdash_\mathcal{S} p : \varsigma
\end{align*}
where $\varsigma \in PType = Res$, and its rules are shown in~\figurename~\ref{fig:imp-S-types}.
Also in this case we have coloured and framed the results of $\hltr{\tr}$ and $\hlcheckjoin{\echeckjoin}$.
\begin{figure}[!tb]
  \scriptsize
    \begin{mathpar}
        \inferrule[($\mathcal{S}$-Const)]
        {\hlcheckjoin{c \in \mathbb{N} \cup \{true, false\}}}
        {\Gamma \vdash_{\mathcal{S}} c : L}

        \inferrule[($\mathcal{S}$-Var)]
        {
            \hlcheckjoin{\Gamma (x) = \tau\,\mathit{var} \land \varsigma = \Gamma (x)}
        }
        {\Gamma \vdash_{\mathcal{S}} x : \varsigma}

        \inferrule[($\mathcal{S}$-Not)]
        {
            \Gamma \vdash_{\mathcal{S}} b : \tau_b\\
            \hlcheckjoin{\tau = \tau_b}
        }
        {
            \Gamma \vdash_{\mathcal{S}} \wop{}{b}{\mathtt{not}} : \tau
        }

        \inferrule[($\mathcal{S}$-Skip)]
        {
            \hlcheckjoin{\varsigma = H\,\mathit{cmd}}
        }
        {
            \Gamma \vdash_{\mathcal{S}} \wskip : \varsigma
        }

        \inferrule[($\mathcal{S}$-Op)]
        {
            \hltr{\Gamma} \vdash_{\mathcal{S}} p_0 : \tau_0\\
            \hltr{\Gamma} \vdash_{\mathcal{S}} p_1 : \tau_1\\
            \hlcheckjoin{op \in \{ \mathtt{+}, \mathtt{*}, \mathtt{-}, \mathtt{or}, \mathtt{\leq}, \ldots \} \land \tau_0 = \tau_1 \land \varsigma = \tau_0}
        }
        {
            \Gamma \vdash_{\mathcal{S}} \wop{p_0}{p_1}{op} : \varsigma
        }

        \inferrule[($\mathcal{S}$-Assign)]
        {
            \hltr{\Gamma} \vdash_{\mathcal{S}} a : \tau_a\\
            \hlcheckjoin{\Gamma(x) = \tau\,\mathit{var} \land \tau = \tau_a \land \varsigma = \tau\,\mathit{cmd}}
        }
        {
            \Gamma \vdash_{\mathcal{S}} \wassign{x}{a} : \varsigma
        }

        \inferrule[($\mathcal{S}$-If)]
        {
            \hltr{\Gamma} \vdash_{\mathcal{S}} b : \tau_b\\
            \hltr{\Gamma} \vdash_{\mathcal{S}} c_1 : \tau_1\,\mathit{cmd}\\
            \hltr{\Gamma} \vdash_{\mathcal{S}} c_2 : \tau_2\,\mathit{cmd}\\
            \hlcheckjoin{\tau_b = \tau_1 = \tau_2 \land \varsigma = \tau_b\,\mathit{cmd}}
        }
        {
            \Gamma \vdash_{\mathcal{S}} \wifthenelse{b}{c_1}{c_2} : \varsigma
        }

        \inferrule[($\mathcal{S}$-While)]
        {
            \hltr{\Gamma} \vdash_{\mathcal{S}} b : \tau_b\\
            \hltr{\Gamma} \vdash_{\mathcal{S}} c_1 : \tau_1\,\mathit{cmd}\\
            \hlcheckjoin{\tau_b = \tau_1 \land \varsigma = \tau_b\,\mathit{cmd}}
        }
        {
            \Gamma \vdash_{\mathcal{S}} \wwhile{b}{c_1} : \varsigma
        }

        \inferrule[($\mathcal{S}$-Seq)]
        {
            \hltr{\Gamma} \vdash_{\mathcal{S}} c_1 : \tau_1\,\mathit{cmd}\\
            \hltr{\Gamma} \vdash_{\mathcal{S}} c_2 : \tau_2\,\mathit{cmd}\\
            \hlcheckjoin{\tau_1 = \tau_2 \land \varsigma = \tau_1\,\mathit{cmd}}
        }
        {
            \Gamma \vdash_{\mathcal{S}} \wseq{c_1}{c_2} : \varsigma
        }

           \inferrule[(SS-Sub)]
        {
            \Gamma \vdash_{\mathcal{S}} p : \varsigma_1\\
            \varsigma_1 \subseteq \varsigma_2
        }
        {
            \Gamma \vdash_{\mathcal{S}} p : \varsigma_2
        }

        \inferrule[(SS-Base)]
        { }
        {
            L \subseteq H
        }

        \inferrule[(SS-Cmd)]
        { \tau' \subseteq \tau }
        {
            \tau\,\mathit{cmd} \subseteq \tau'\,\mathit{cmd}
        }

        \inferrule[(SS-Refl)]
        { }
        {
            \varsigma \subseteq \varsigma
        }

        \inferrule[(SS-Tr)]
        {
            \varsigma_1 \subseteq \varsigma_2\\
            \varsigma_2 \subseteq \varsigma_3
        }
        {
            \varsigma_1 \subseteq \varsigma_3
        }

 \end{mathpar}

    \caption{The rules of the type checking algorithm $\mathcal{S}$ (with subtyping) for {\footnotesize\texttt{WHILE}}.}\label{fig:imp-S-types}
\end{figure}
%
In the following we assume that the initial typing environment $\Gamma$ contains the security level of each variable occurring in the program at hand.

\paragraph{\stepone}
The shape of the caches is as expected:
\[
    C \in Cache = \wp(Phrase \times Env \times PType)
\]

\paragraph{\steptwo}
The function $\ebuildcache$ is defined in~\figurename~\ref{fig:imp-buildcache}.
\begin{figure}[htb!]
\scriptsize
\begin{align*}
    \buildcache{(c : L)}{\Gamma} &\triangleq \{ (c, \emptyset, L) \} \qquad c \in \mathbb{N} \cup \{ true, false\}\\
    \buildcache{(x : \tau)}{\Gamma} &\triangleq \{ (x, [x \mapsto \tau\,\mathit{var}], \tau) \}\\
    \buildcache{(\wop{a_1}{a_2}{op} : \tau)}{\Gamma} &\triangleq \{ (\wop{a_1}{a_2}{op}, \restrict{\Gamma}{\mathit{FV}(\wop{a_1}{a_2}{op})}, \tau) \} \\
    &\qquad\cup (\buildcache{(a_1 : \tau_1)}{\hltr{\Gamma}})\\
    &\qquad\cup (\buildcache{(a_2 : \tau_2)}{\hltr{\Gamma}})\\
    \buildcache{(\wop{a_1}{a_2}{\leq} : \tau)}{\Gamma} &\triangleq \{ (\wop{a_1}{a_2}{\leq}, \restrict{\Gamma}{\mathit{FV}(\wop{a_1}{a_2}{\leq})}, \tau) \} \\
    &\qquad\cup (\buildcache{(a_1 : \tau_1)}{\hltr{\Gamma}})\\
    &\qquad\cup (\buildcache{(a_2 : \tau_2)}{\hltr{\Gamma}})\\
    \buildcache{(\wop{b_1}{b_2}{or} : \tau)}{\Gamma} &\triangleq \{ (\wop{b_1}{b_2}{or}, \restrict{\Gamma}{\mathit{FV}(\wop{b_1}{b_2}{or})}, \tau) \} \\
    &\qquad\cup (\buildcache{(b_1 : \tau_1)}{\hltr{\Gamma}})\\
    &\qquad\cup (\buildcache{(b_2 : \tau_2)}{\hltr{\Gamma}})\\
    \buildcache{(\wop{\!}{b}{not} : \tau)}{\Gamma} &\triangleq \{ (\wop{}{b}{not}, \restrict{\Gamma}{\mathit{FV}(\wop{}{b}{not})}, \tau) \}\\
    &\qquad\cup (\buildcache{ (b : \tau)}{\hltr{\Gamma}})\\
    \buildcache{(skip : H\,\mathit{cmd})}{\Gamma} &\triangleq \{ (skip, \emptyset, H\,\mathit{cmd}) \}\\
    \buildcache{(\wassign{x}{a} : \tau\,\mathit{cmd})}{\Gamma} &\triangleq \{ (\wassign{x}{a}, \restrict{\Gamma}{\mathit{FV}(\wassign{x}{a})}, \tau\,\mathit{cmd}) \}\\
    &\qquad\cup(\buildcache{(x : \tau_x)}{\hltr{\Gamma}}) \\
    &\qquad\cup (\buildcache{(a : \tau_a)}{\hltr{\Gamma}})\\
    \buildcache{(\wifthenelse{b}{c_1}{c_2} : \tau\,\mathit{cmd})}{\Gamma} &\triangleq \{ (\wifthenelse{b}{c_1}{c_2}, \restrict{\Gamma}{\mathit{FV}(\wifthenelse{b}{c_1}{c_2})}, \tau\,\mathit{cmd}) \}\\
    &\qquad\cup(\buildcache{(b : \tau_b)}{\hltr{\Gamma}}) \\
    &\qquad\cup(\buildcache{(c_1 : \tau_1\,\mathit{cmd})}{\hltr{\Gamma}}) \\
    &\qquad\cup(\buildcache{(c_2 : \tau_2\,\mathit{cmd})}{\hltr{\Gamma}}) \\
    \buildcache{(\wwhile{b}{c} : \tau\,\mathit{cmd})}{\Gamma} &\triangleq \{ (\wwhile{b}{c}, \restrict{\Gamma}{\mathit{FV}(\wwhile{b}{c})}, \tau\,\mathit{cmd}) \}\\
    &\qquad\cup(\buildcache{(b : \tau_b)}{\hltr{\Gamma}}) \\
    &\qquad\cup(\buildcache{(c : \tau_c\,\mathit{cmd})}{\hltr{\Gamma}}) \\
    \buildcache{(\wseq{c_1}{c_2} : \tau\,\mathit{cmd})}{\Gamma} &\triangleq \{ (\wseq{c_1}{c_2}, \restrict{\Gamma}{\mathit{FV}(\wseq{c_1}{c_2})}, \tau\,\mathit{cmd}) \}\\
    &\qquad\cup(\buildcache{(c_1 : \tau_1\,\mathit{cmd})}{\hltr{\Gamma}}) \\
    &\qquad\cup(\buildcache{(c_2 : \tau_2\,\mathit{cmd})}{\hltr{\Gamma}})
\end{align*}
\caption{Definition of $\ebuildcache$ for the incremental type checking of {\footnotesize\texttt{WHILE}}.}\label{fig:imp-buildcache}
\vspace{-5mm}
\end{figure}
%

\paragraph{\stepthree}
In \figurename~\ref{fig:imp-S-incremental} we display the rules defining the algorithm $\mathcal{IS}$ with judgements of the following form
\[
    \Gamma, C \vdash_\mathcal{IS} p : \varsigma \triangleright C'
\]
As expected, most of the rules are trivial instantiations of rules in Section~\ref{sec:foundations} that mimic those of the original type checking algorithm.
Of course, $\mathcal{IS}$ inherits unchanged the subtyping relation of $\mathcal{S}$.
\begin{figure}[!htbp]
    \scriptsize
    \begin{mathpar}
        \inferrule*[lab={($\mathcal{IS}$-Hit)}]
        {
            C(p) = \langle \Gamma', \varsigma \rangle\\
            \compatenv{\Gamma}{\Gamma'}{p}
        }
        {
            \Gamma, C \vdash_\mathcal{IS} p : \varsigma \triangleright C
        }
        \quad
        \inferrule*[lab={($\mathcal{IS}$-Const-Miss)},right={$\mathit{miss} (C, c, \Gamma)$}]%
        {
            c \in \mathbb{N} \cup \{ true, false \}\\
            \emptyset \vdash_\mathcal{S} c : \varsigma\\
            C' = C \cup \{ (c, \emptyset, \varsigma) \}
        }
        {
            \Gamma, C \vdash_\mathcal{IS} c : \varsigma \triangleright C'
        }

        \inferrule*[lab={($\mathcal{IS}$-Var-Miss)},right={$\mathit{miss}(C, x, \Gamma)$}]%
        {
            \Gamma \vdash_\mathcal{S} c : \varsigma\\
            \hlcheckjoin{\varsigma = \tau\,\mathit{var}}\\
            C' = C \cup \{ (x, \restrict{\Gamma}{x}, \varsigma) \}
        }
        {
            \Gamma, C \vdash_\mathcal{IS} x : \varsigma \triangleright C'
        }
        \quad
        \inferrule*[lab={($\mathcal{IS}$-Skip-Miss)},right={$\mathit{miss} (C, \wskip, \Gamma)$}]%
        {
            \hltr{\Gamma} \vdash \wskip : \varsigma \\
            C' = C \cup \{ (\wskip, \emptyset, \varsigma) \}
        }
        {
            \Gamma, C \vdash_\mathcal{IS} \wskip : \varsigma \triangleright C'
        }

        \inferrule*[lab={($\mathcal{IS}$-Op-Miss)},right={$\mathit{miss} (C, \wop{a_1}{a_2}{op}, \Gamma)$}]%
        {
            \hltr{\Gamma}, C \vdash_\mathcal{IS} a_1  : \tau_1 \triangleright C'' \\
            \hltr{\Gamma}, C \vdash_\mathcal{IS} a_2  : \tau_2 \triangleright C''' \\
            \hlcheckjoin{\tau_1 = \tau_2 \land \varsigma = \tau_1}\\
            C' = C'' \cup C''' \cup \{ (\wop{a_1}{a_2}{op}, \restrict{\Gamma}{\mathit{FV}(\wop{a_1}{a_2}{op})}, \varsigma) \}
        }
        {
            \Gamma, C \vdash_\mathcal{IS} \wop{a_1}{a_2}{op}  : \varsigma \triangleright C'
        }

        \inferrule*[lab={($\mathcal{IS}$-BOp-Miss)},right={$\mathit{miss} (C, \wop{b_1}{b_2}{or}, \Gamma)$}]%
        {
            \hltr{\Gamma}, C \vdash_\mathcal{IS} b_1  : \tau_1 \triangleright C'' \\
            \hltr{\Gamma}, C \vdash_\mathcal{IS} b_2  : \tau_2 \triangleright C''' \\
            \hlcheckjoin{\tau_1 = \tau_2 \land \varsigma = \tau_1} \\
            C' = C'' \cup C''' \cup \{ (\wop{b_1}{b_2}{or}, \restrict{\Gamma}{\mathit{FV}(\wop{b_1}{b_2}{or})}, \varsigma) \}
        }
        {
            \Gamma, C \vdash_\mathcal{IS} \wop{b_1}{b_2}{or}  : \varsigma \triangleright C'
        }

        \inferrule*[lab={($\mathcal{IS}$-Not-Miss)},right={$\mathit{miss} (C, \wop{}{b}{not}, \Gamma)$}]%
        {
            \hltr{\Gamma}, C \vdash_\mathcal{IS} b : \tau \triangleright C'' \\
            C' = C'' \cup \{ (\wop{}{b}{not}, \restrict{\Gamma}{\mathit{FV}(\wop{}{b}{not})}, \tau) \}
        }
        {
            \hltr{\Gamma}, C \vdash_\mathcal{IS} \wop{}{b}{not}  : \tau \triangleright C'
        }

        \inferrule*[lab={($\mathcal{IS}$-Leq-Miss)},right={$\mathit{miss} (C, \wop{a_1}{a_2}{\leq}, \Gamma)$}]%
        {
            \hltr{\Gamma}, C \vdash_\mathcal{IS} a_1  : \tau_1 \triangleright C'' \\
            \hltr{\Gamma}, C \vdash_\mathcal{IS} a_2  : \tau_2 \triangleright C''' \\
            \hlcheckjoin{\tau_1 = \tau_2 \land \varsigma = \tau_1}\\
            C' = C'' \cup C''' \cup \{ (\wop{a_1}{a_2}{\leq}, \restrict{\Gamma}{\mathit{FV}(\wop{a_1}{a_2}{\leq})}, \varsigma) \}
        }
        {
            \Gamma, C \vdash_\mathcal{IS} \wop{a_1}{a_2}{\leq}  : \varsigma \triangleright C'
        }

        \inferrule*[lab={($\mathcal{IS}$-Assign-Miss)},right={$\mathit{miss} (C, \wassign{x}{a}, \Gamma)$}]%
        {
            \hltr{\Gamma}, C \vdash_\mathcal{IS} x  : \tau_x\,\mathit{var} \triangleright C'' \\
            \hltr{\Gamma}, C \vdash_\mathcal{IS} a  : \tau_a \triangleright C''' \\
            \hlcheckjoin{\tau_a = \tau_x \land \varsigma = \tau_a\,\mathit{cmd}}\\
            C' = C'' \cup C''' \cup \{ (\wassign{x}{a}, \restrict{\Gamma}{\mathit{FV}(\wassign{x}{a})}, \varsigma) \}
        }
        {
            \Gamma, C \vdash_\mathcal{IS} \wassign{x}{a} : \varsigma \triangleright C'
        }

        \inferrule*[lab={($\mathcal{IS}$-If-Miss)},right={$\mathit{miss} (C, \wifthenelse{b}{c_1}{c_2}, \Gamma)$}]%
        {
            \hltr{\Gamma}, C \vdash_\mathcal{IS} b  : \tau_b \triangleright C'' \\
            \hltr{\Gamma}, C \vdash_\mathcal{IS} c_1  : \tau_1\,\mathit{cmd} \triangleright C''' \\
            \hltr{\Gamma}, C \vdash_\mathcal{IS} c_2  : \tau_2\,\mathit{cmd} \triangleright C^{iv} \\
            \hlcheckjoin{\tau_1 = \tau_2 = \tau_b \land \varsigma = \tau_1\,\mathit{cmd}}\\
            C' = C'' \cup C''' \cup C^{iv} \cup \{ (\wifthenelse{b}{c_1}{c_2}, \restrict{\Gamma}{\mathit{FV}(\wifthenelse{b}{c_1}{c_2})}, \varsigma) \}
        }
        {
            \Gamma, C \vdash_\mathcal{IS} \wifthenelse{b}{c_1}{c_2}  : \varsigma \triangleright C'
        }

        \inferrule*[lab={($\mathcal{IS}$-While-Miss)},right={$\mathit{miss} (C, \wwhile{b}{c}, \Gamma)$}]%
        {
            \hltr{\Gamma}, C \vdash_\mathcal{IS} b  : \tau_b \triangleright C'' \\
            \hltr{\Gamma}, C \vdash_\mathcal{IS} c  : \tau_1\,\mathit{cmd} \triangleright C''' \\
            \hlcheckjoin{\tau_1 = \tau_b \land \varsigma = \tau_1\,\mathit{cmd}}\\
            C' = C'' \cup C''' \cup \{ (\wwhile{b}{c}, \restrict{\Gamma}{\mathit{FV}(\wwhile{b}{c})}, \varsigma) \}
        }
        {
            \Gamma, C \vdash_\mathcal{IS} \wwhile{b}{S} : \varsigma \triangleright C'
        }

        \inferrule*[lab={($\mathcal{IS}$-Seq-Miss)},right={$\mathit{miss} (C, \wseq{c_1}{c_2}, \Gamma)$}]%
        {
            \hltr{\Gamma}, C \vdash_\mathcal{IS} c_1  : \tau_1\,\mathit{cmd} \triangleright C'' \\
            \hltr{\Gamma}, C \vdash_\mathcal{IS} c_2  : \tau_2\,\mathit{cmd} \triangleright C''' \\
            \hlcheckjoin{\tau_1 = \tau_2 \land \varsigma = \tau_1\,\mathit{cmd}}\\
            C' = C'' \cup C''' \cup \{ (\wseq{c_1}{c_2}, \restrict{\Gamma}{\mathit{FV}(\wseq{c_1}{c_2})}, \varsigma) \}
        }
        {
            \Gamma, C \vdash_\mathcal{IS} \wseq{c_1}{c_2} : \varsigma \triangleright C'
        }
    \end{mathpar}

    \begin{equation}\label{def:compS}
        \text{with}\,\,\compatenv{\Gamma}{\Gamma'}{p} \triangleq dom(\Gamma) \supseteq \mathit{FV}(p) \land dom(\Gamma') \supseteq \mathit{FV}(p) \land \forall y \in \mathit{FV}(p) \,.\, \Gamma(y) = \Gamma'(y)
    \end{equation}
    \caption{Rules defining incremental algorithm $\mathcal{IS}$ to type check \texttt{WHILE}.}\label{fig:imp-S-incremental}
\end{figure}

\paragraph{\stepfour}
Also in this case the type coherence of algorithm $\mathcal{IS}$ follows from the fact that $\ecompatenv$ satisfies Definition~\ref{def:tec}.
\begin{restatable}{lemma}{lemmaIScoherence}\label{lemma:IS-coherence}
The predicate
    $\ecompatenv$ of Eq.~(\ref{def:compS}) in \figurename~\ref{fig:imp-S-incremental} expresses compatibility.
\end{restatable}
%
%
Now we have the following theorem, again an instance of Theorem~\ref{thm:typecoherence}.
\begin{restatable}{theorem}{thmIScoherence}\label{thm-IS-coherence}
$ \quad
\forall \Gamma, C, e.\
\Gamma    \vdash_\mathcal{S} e : \tau  \iff
\Gamma, C \vdash_\mathcal{IS} e : \tau \triangleright C'
$
\end{restatable}


\section{Implementation and some experiments}\label{sec:exp}
%

We have implemented in OCaml our proposal making incremental the usage the type-checker of MinCaml~\cite{MinCaml}.\footnote{Available at \url{https://github.com/mcaos/incremental-mincaml}}
It was enough wrapping it as dictated by the formal definitions of Section~\ref{sec:foundations}.
In detail, caches and type environments are implemented as hash-tables, so their handling is done almost in constant time.
The memory overhead due to the cache is $\mathcal{O}(n \times m)$, where $n$ is the size of the program under analysis and $m$ is the number of variables therein.
The other possible time consuming part concerns checking environment compatibility.
The key idea to make $\ecompatenv$ efficient is to compute the sets of the free variables beforehand, and to store them as additional annotations on the aAST.
Summing up, implementing our schema is not too demanding, since it can be done with standard data structures.

%


Next, we show that
\begin{enumerate*}[label=(\roman*)]
\item the cost of using the type checker incrementally depends on the size of \emph{diffs};
\item its performance increases as these become smaller; and
\item the incremental usage is almost always faster than re-using the standard one.
\end{enumerate*}
The comparison is done by type checking synthetic programs with (binary and complete) aAST of increasing depth from 8 to 16, and with a number of variables ranging from 1 to $2^{15}$.
All the internal nodes are binary operators and the leaves are free variables.
This test suites are intended to stress our incremental algorithm in the worst, yet artificial case.
The measures are obtained using the library \emph{Benchmark} that takes into account the overhead of OCaml runtime.\footnote{%
Available at \url{https://github.com/Chris00/ocaml-benchmark}}

To test the efficiency of caching we first re-typed twice the program with no change, starting with an empty cache.
Table~\ref{tab:identity} displays the number of re-typings per second in function of the depth of the aAST and the number of variables in the program.
\begin{table}[bt]
    \caption{Experimental results about caching in terms of tree re-checks per second for the standard and the incremental usage of the type checker.}\label{tab:identity}
    \begin{center}
        \begin{tabular}{llrrr}
        \toprule
        {Depth\ \ \ \ } & {Vars} & {\ \ \ Standard} & {\ \ \ Incremental} & {\ \ \ Incremental$/$Standard ratio} \\
        \midrule
        $16$ & $1$ & $84.77$ & $206492.96$ & $2435.93$\\
        $16$ & $2^{7}$ & $72.65$ & $6642.75$ & $91.44$\\
        $16$ & $2^{9}$ & $72.20$ & $1417.95$ & $19.64$\\
        $16$ & $2^{11}$ & $68.01$ & $373.61$ & $5.49$\\
        $16$ & $2^{13}$ & $60.33$ & $96.05$ & $1.59$\\
        $16$ & $2^{15}$ & $57.96$ & $34.75$ & $0.60$\\
        \bottomrule
        \end{tabular}
        \end{center}
  \end{table}
Clearly, the overhead for caching is largely acceptable -- and caching is also beneficial when the number of free variables is not too large w.r.t.~the aAST depth because the results of common subtrees are re-used.

Then, we have simulated program changes by invalidating parts of caches that correspond to the rightmost subexpression at different depths.
Note that invalidating cache entries for the \emph{diff} subexpression $e'$ of $e$ requires to invalidate
\begin{enumerate*}[label=(\roman*)]
\item all the entries for the nodes in the path from the root of the aAST of $e$ to $e'$ and
\item all the entries for $e'$ and its subexpressions, recursively.
\end{enumerate*}
The plots in Figure~\ref{fig:diff-cmp} represent the number of re-typings per second vs. the size of the \emph{diff} for a few choices of aAST depth and number of variables.
However, the shape of the curves is essentially the same also for different values, as shown by the additional diagrams in the Appendix (more are available in the GitHub repository).
\begin{figure}[tb]
  \centering
  \begin{subfigure}[b]{0.47\textwidth}
      \includegraphics[width=\textwidth]{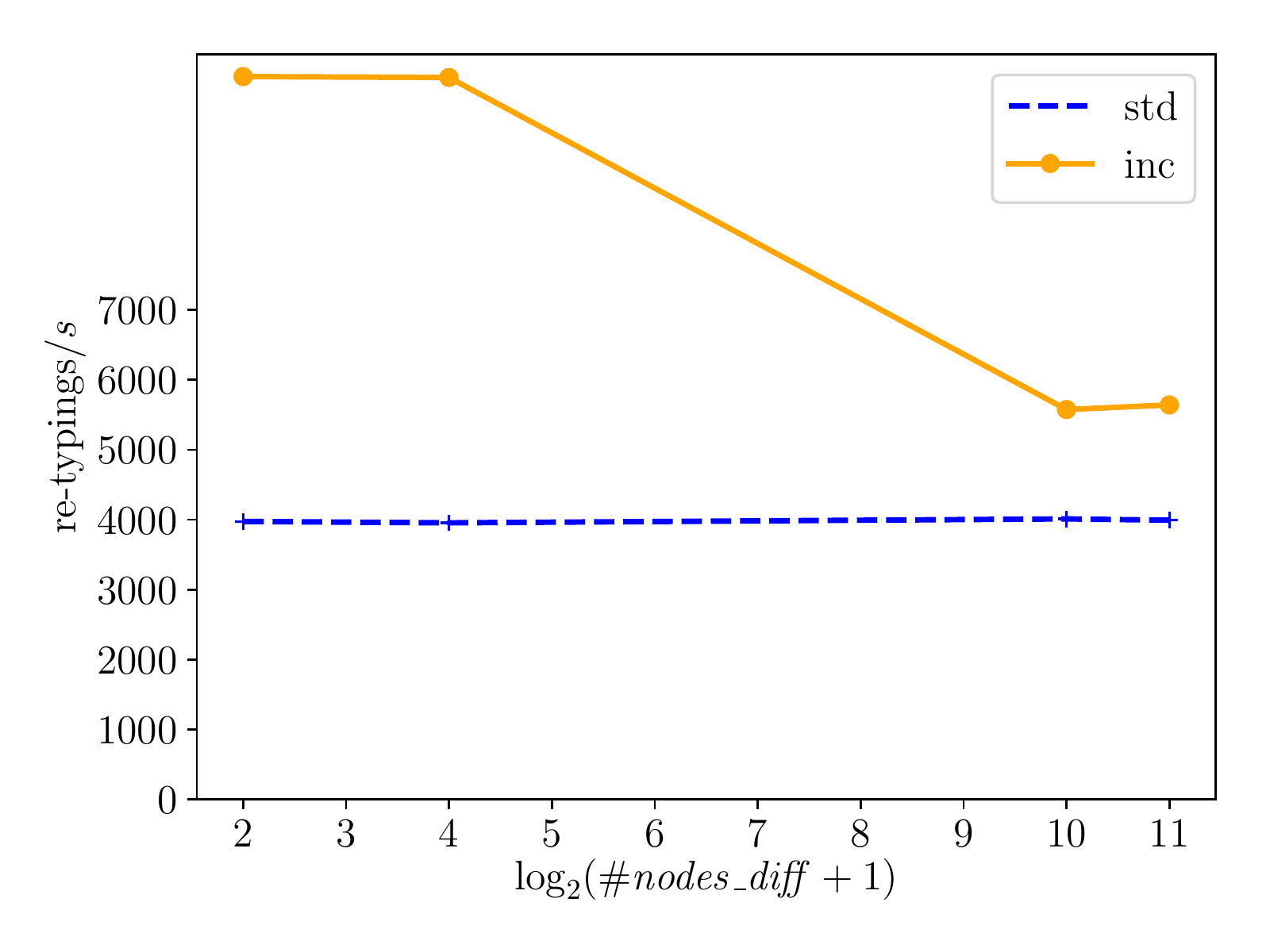}
      \caption{$\mathit{depth} = 12, \mathit{\#variables} = 2^{9}$}
      \label{fig:plot-orig-einc-12a}
  \end{subfigure}
  \hspace{1em}
  \begin{subfigure}[b]{0.47\textwidth}
    \includegraphics[width=\textwidth]{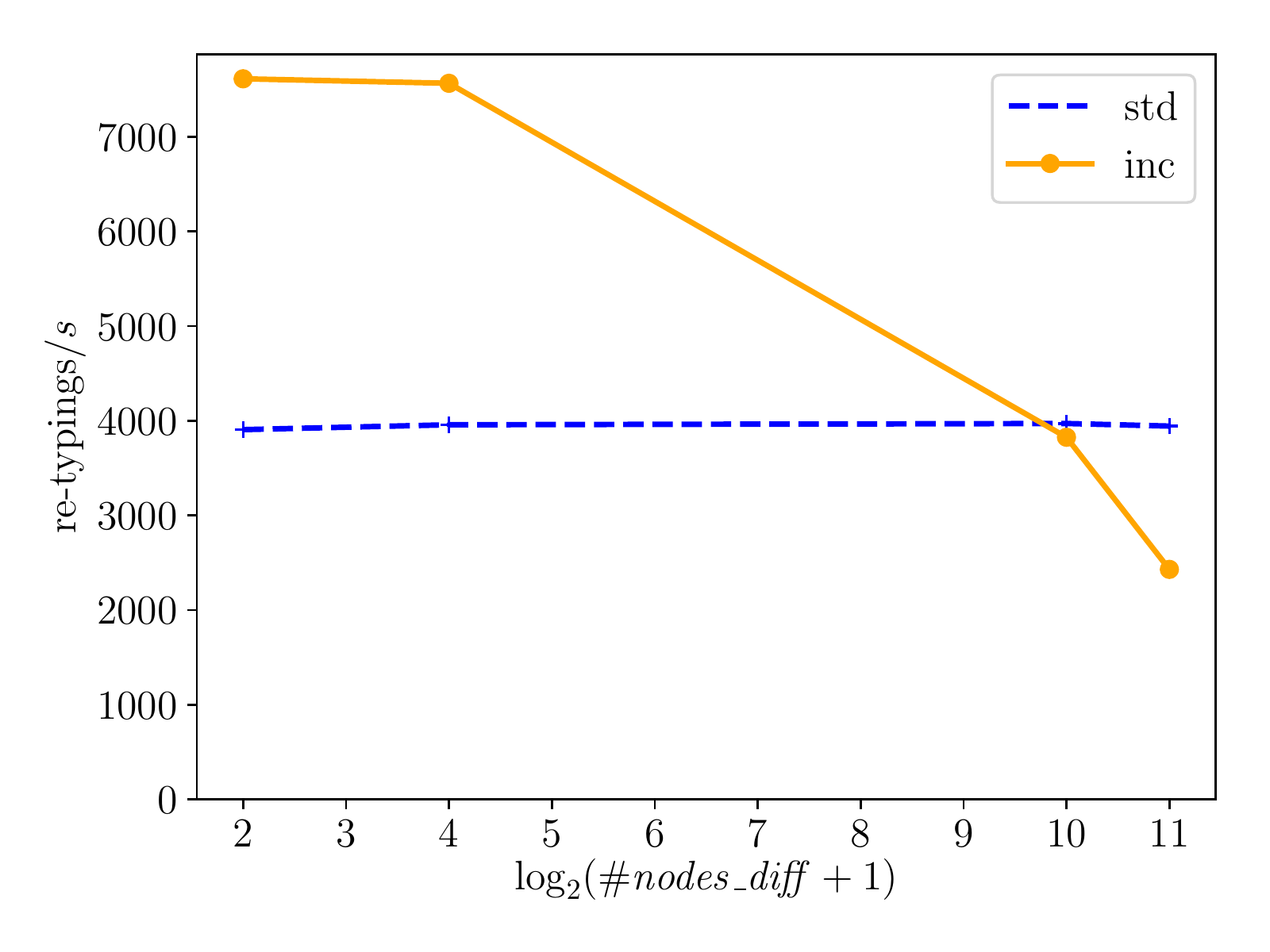}
    \caption{$\mathit{depth} = 12, \mathit{\#variables} = 2^{11}$}
    \label{fig:plot-orig-einc-12b}
  \end{subfigure}
    \begin{subfigure}[b]{0.47\textwidth}
    \includegraphics[width=\textwidth]{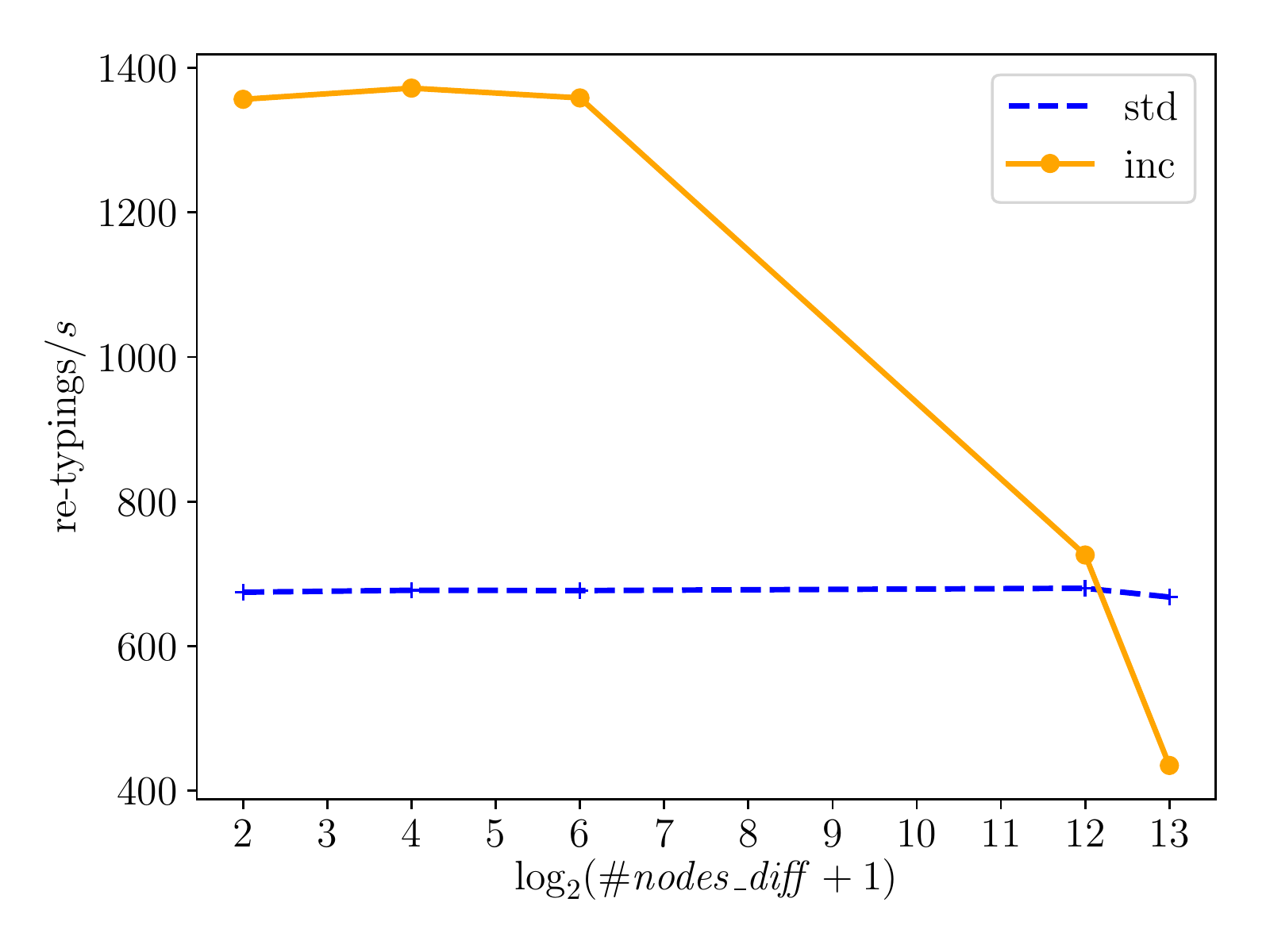}
    \caption{$\mathit{depth} = 14, \mathit{\#variables} = 2^{13}$}
    \label{fig:plot-orig-einc-14}
  \end{subfigure}
  \hspace{1em}
  \begin{subfigure}[b]{0.47\textwidth}
    \includegraphics[width=\textwidth]{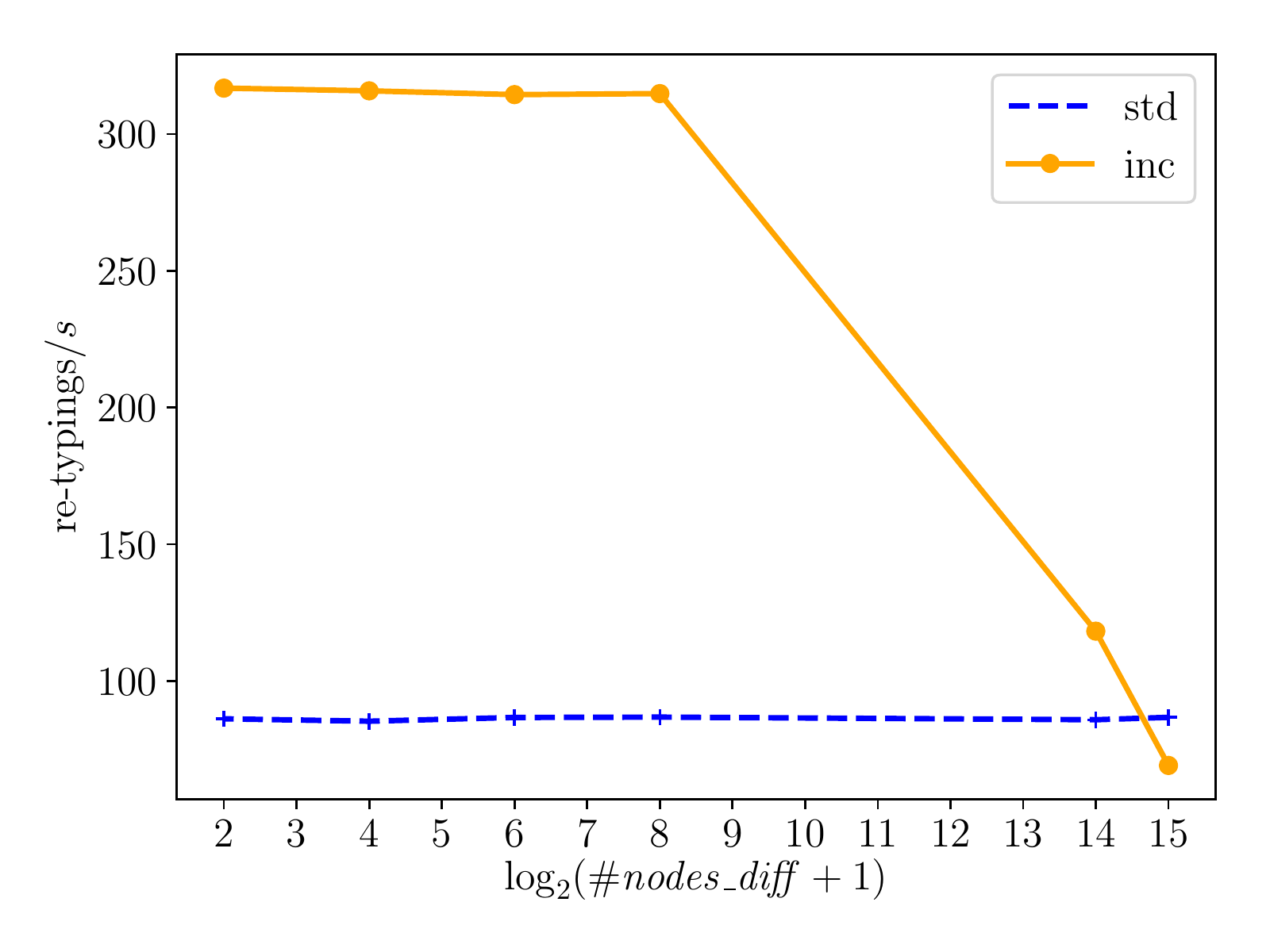}
    \caption{$\mathit{depth} = 16, \mathit{\#variables} = 2^{15}$}
    \label{fig:plot-orig-einc-16}
  \end{subfigure}
  \caption{Experimental results comparing the number of re-typings per second vs. the number of nodes of \emph{diff}.
  The blue, dashed plot is for the standard type checking, while the orange, solid one is for the incremental usage.
  The $x$-axis is logarithmic in all sub-tables, while the $y$-axis is scaled as necessary.}\label{fig:diff-cmp}
\end{figure}
The experimental results show that our caching and memoization is faster than re-typing twice.
An exception is when aAST have the maximum number of variables and the considered changes exceed 25\% of the nodes;
this is shown in the rightmost parts of \figurename~\ref{fig:diff-cmp}(b-d), while in part (a) the number of variables is not the maximum possible and incrementality is always beneficial.
%
%
All in all, the advantage of using incrementally a type checker decreases, as expected, when there is a significant growth of the number of variables or in the size of the program.
However, these cases only show up with very big numbers, which are not likely to occur often.


\section{Conclusions}\label{sec:concl}
%

We have presented an algorithmic schema for incrementally using existing type checking and type inference algorithms.
Since only the shape of the input, the output, and some domain-specific knowledge of the original algorithms are relevant, our schema considers them as grey-boxes.
Remarkably, the only real effort for defining the incremental algorithm is required for establishing the notion of compatibility between parts of the environments relevant for re-typing.
We have introduced the basic bricks of our approach and proved a theorem guaranteeing the coherence of \emph{any} original algorithm with its incremental version, and \emph{vice versa}.
As a matter of fact, coherence follows from easily checking a mild condition on the environment compatibility.
To illustrate the approach we have then instantiated our proposal to a functional language for type checking and type inference, and to an imperative language for checking non-interference.

We have implemented the incremental version of the type checker of MinCaml,
and we have assessed it on synthetic programs with varying size and number of variables.
The experiments have shown our proposal worth using within a continuous software development model where fast responsiveness
is needed.
Indeed, the diagrams in \figurename~\ref{fig:diff-cmp} show that only \emph{diffs} are typed, possibly with those parts of the code affected by them.
Additionally, the cost of using the type checker incrementally depends on the size of \emph{diffs}, and its performance increases as these become smaller, a typical situation when applying local transformations, e.g.\ code motion, dead code elimination, and code wrapping.
%

\paragraph{Future work.}
We are confident that little extensions to our proposal are needed to cover also type and effect systems.
Also, other programming paradigms should be easily accommodated in our incremental schema, as preliminary results on process calculi suggest us.
More work is instead required to apply our ideas to other syntax-directed static analyses, e.g. control flow analysis
because of fixed-point computations.
We also plan to carry our proposal on Abstract Interpretation, where the rich structure of the abstract domains poses some serious challenges.
Presently, we are extending our prototype with an incremental type inference for MinCaml.
Moreover, we plan to implement a generator that, given an existing type checking or inference algorithm $\mathcal A$ and the definition of $\ecompatenv$ automatically produces the corresponding incremental algorithm  (recall that $\tr$ and $\echeckjoin$ are directly inherited from $\mathcal A$).

More experiments on real programs are also in order to better assess the performance of our proposal, as well as its scalability.


\newpage

\bibliographystyle{splncs04}
\bibliography{biblio}
\FloatBarrier

\newpage
\appendix
\section{Proofs of Lemmata and Theorems}\label{app:proof}
%

\thmcachecorrectness*
\begin{proof}
    The theorem easily follows by induction from the definition of $\ebuildcache$.
\end{proof}

\thmtypecoherence*
\begin{proof}
The proof is divided in two parts, one for each side of the implication.

\paragraph{$(\implies)$}

We show by structural induction on terms that, under the hypothesis of the theorem, it is true that
\[
    \Gamma \vdash_\mathcal{A} t : R \implies \Gamma, C \vdash_\mathcal{IA} t : R \triangleright C'.
\]

\begin{description}
    \item[Base case]
    This case occurs when $t$ has no subterms, and it has two exhaustive sub-cases:
    \begin{enumerate}
        \item If $\mathit{miss}(C, t, \Gamma)$, then the premises for the miss rule for terms with no subterms are trivially satisfied and we can derive $\Gamma, C \vdash_\mathcal{IA} t : R \triangleright C'$.
        \item Otherwise, $C(t) = \langle \Gamma', R \rangle \,\land\, \compatenv{\Gamma}{\Gamma'}{t}$ holds, and $\Gamma, C \vdash_\mathcal{IA} t : R \triangleright C$ follows from the hit rule.
    \end{enumerate}

    \item[Inductive case.]
    Assume that for any subterm of $t$ the implication holds.    
    Again, we distinguish two cases:
    \begin{enumerate}
        \item If $\mathit{miss}(C, t, \Gamma)$, then since we know that:
        \begin{itemize}
            \item by the induction hypothesis for any $i \in \mathbb{I}_{t}$, it holds that $tr^{t}_{t_i} (\Gamma, \{ R_j \}_{j \leq i \land j \in \mathbb{I}_{t}}) \vdash_\mathcal{A} t_i : R_i \implies tr^{t}_{t_i} (\Gamma, \{ R_j \}_{j \leq i \land j \in \mathbb{I}_{t}}), C \vdash_\mathcal{IA} t_i : R_i \triangleright C^{i}$
            \item $\echeckjoin_{t} (\Gamma, \{ R_i \}_{i \in \mathbb{I}_{t}}, \mathtt{out} R)$ holds in the premise of the original rule, it holds in the premise of incremental rule too.
        \end{itemize}
        all the premises of the miss rule are satisfied, and $\Gamma, C \vdash_\mathcal{IA} t : R \triangleright C'$ holds. 
        \item Otherwise, apply the same argument of case 2. above. 
    \end{enumerate}
\end{description}

\paragraph{$(\impliedby)$}

Again, we use structural induction terms to show that, under the hypothesis of the theorem, the implication
\[
    \Gamma, C \vdash_\mathcal{IA} t : R \triangleright C' \implies \Gamma \vdash_\mathcal{A} t : R
\]
holds.

\begin{description}
    \item[Base case]
    This case occurs when $t$ has no subterms, and it has two exhaustive sub-cases:
    \begin{enumerate}
        \item If $\mathit{miss}(C, t, \Gamma)$, then $\Gamma \vdash_\mathcal{A} t : R$ is trivially true because it is a premise of the relevant miss rule.
        \item Otherwise $C(t) = \langle \Gamma', R \rangle \,\land\, \compatenv{\Gamma}{\Gamma'}{t}$ must hold, and we can deduce $\Gamma \vdash_\mathcal{A} t : R$, because $\ecompatenv$ expresses compatibility and Theorem~\ref{thm:cachecorrectness} holds.
    \end{enumerate}

    \item[Inductive case.] Just the symmetric of the other implication.
\end{description}
\end{proof}

\lemmaIFcoherence*
\begin{proof}
    Trivial since $\compatenv{\Gamma}{\Gamma'}{e}$ requires $\Gamma$ and $\Gamma'$ to coincide on the free variables of $e$ on which the typing of $e$ only depends.
\end{proof}

\thmIFcoherence*
\begin{proof}
    Immediate by Lemma~\ref{lemma:IF-coherence}.
\end{proof}

\lemmaIWcoherence*
\begin{proof}
Since $\mathcal W$ is syntax-driven, (both the tree and the rules in) a deduction, if any, only depends on $e$.
The implication in Definition~\ref{def:tec} holds because all the premises that use $\Gamma$ still hold when $\Gamma'$ is used instead, since $\Gamma(y)$ unifies with $\Gamma'(y)$ for all free variables $y$ (note that $\mathcal U$ is reflexive, symmetric and transitive).
\end{proof}

\thmIWcoherence*
\begin{proof}
    Immediate by Lemma~\ref{lemma:IW-coherence}.
\end{proof}

\lemmaIScoherence*
\begin{proof}
    Immediate
\end{proof}

\thmIScoherence*
\begin{proof}
    Immediate by Lemma~\ref{lemma:IS-coherence}.
\end{proof}

\newpage

\section*{Further experimental results}\label{sec:more-results}
%

\begin{figure}[!hbp]
  \centering
    \begin{subfigure}[b]{0.47\textwidth}
      \includegraphics[width=\textwidth]{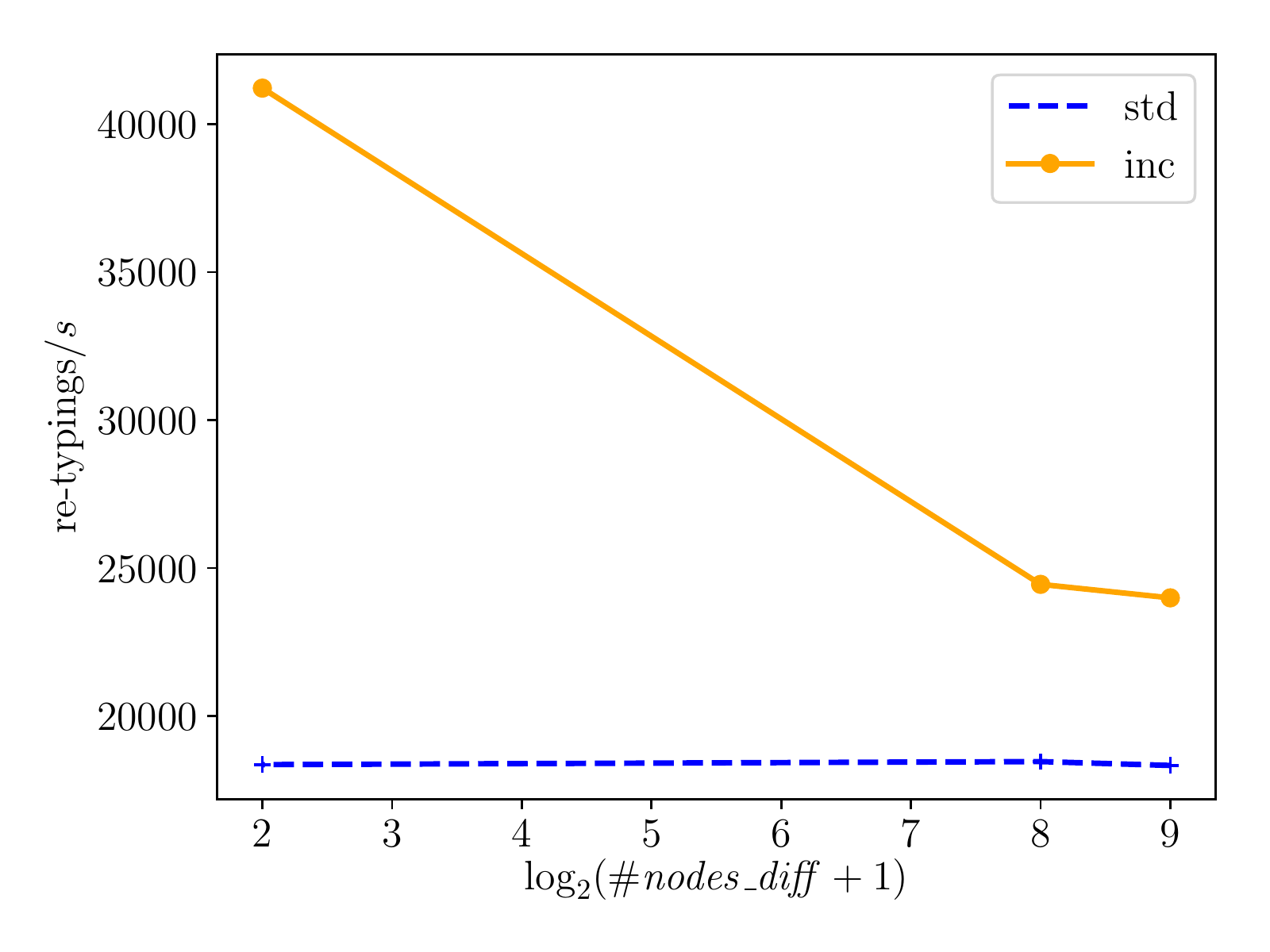}
      \caption{$\mathit{depth} = 10, \mathit{\#variables} = 2^7$}
  \end{subfigure}
  \hspace{1em}
  \begin{subfigure}[b]{0.47\textwidth}
    \includegraphics[width=\textwidth]{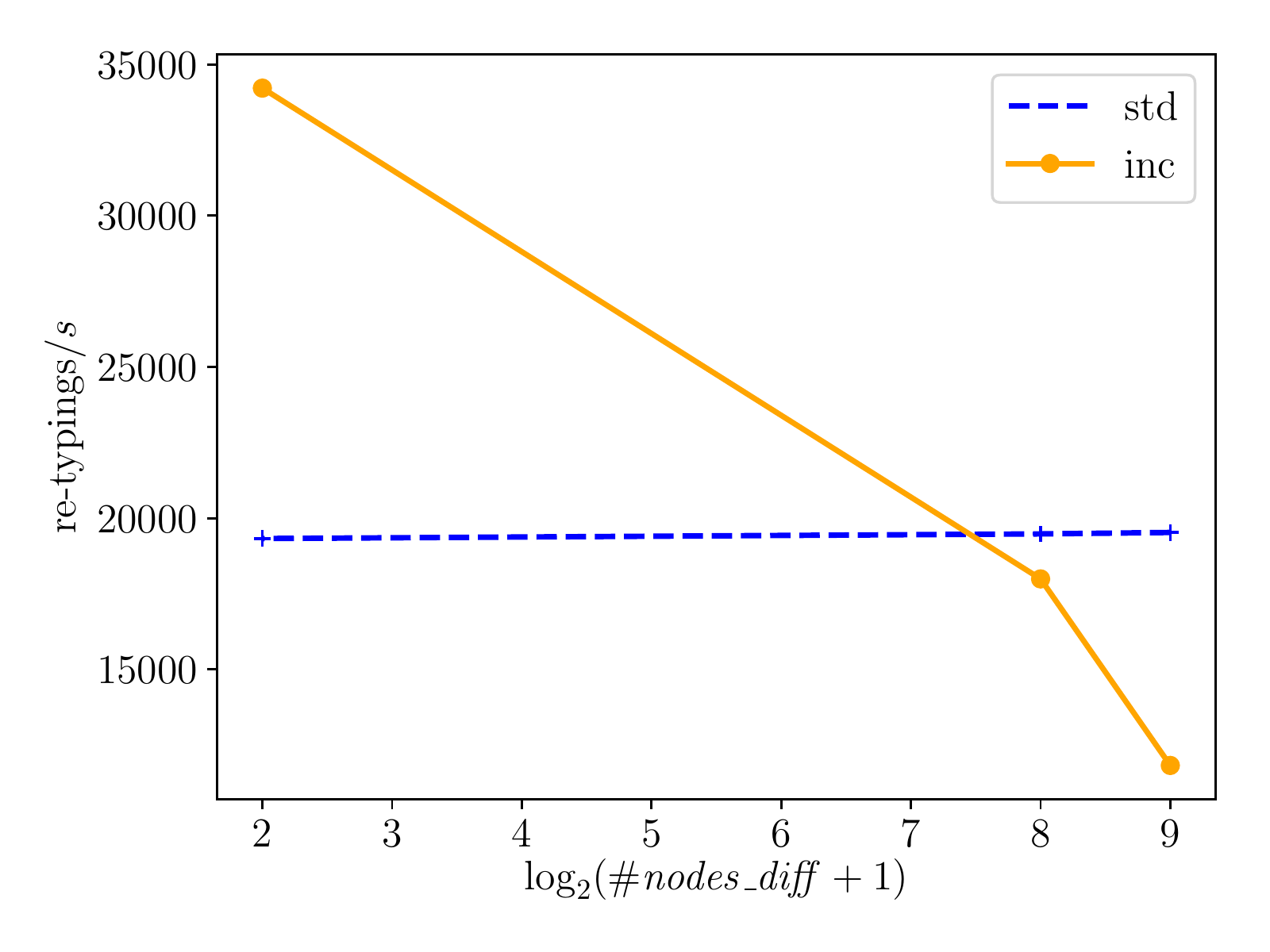}
    \caption{$\mathit{depth} = 10, \mathit{\#variables} = 2^9$}
  \end{subfigure}
  \begin{subfigure}[b]{0.47\textwidth}
      \includegraphics[width=\textwidth]{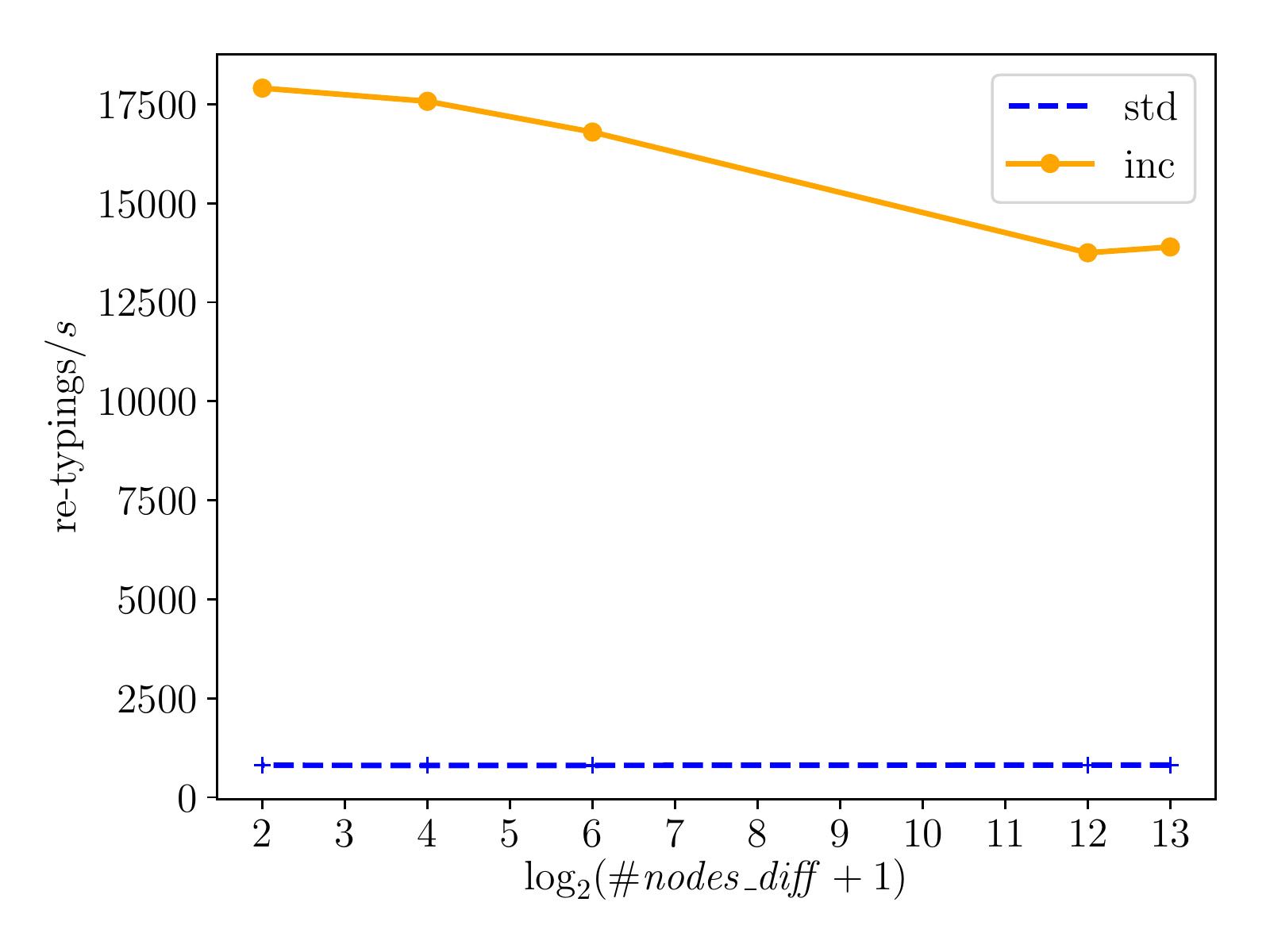}
      \caption{$\mathit{depth} = 14, \mathit{\#variables} = 2^7$}
  \end{subfigure}
  \hspace{1em}
  \begin{subfigure}[b]{0.47\textwidth}
    \includegraphics[width=\textwidth]{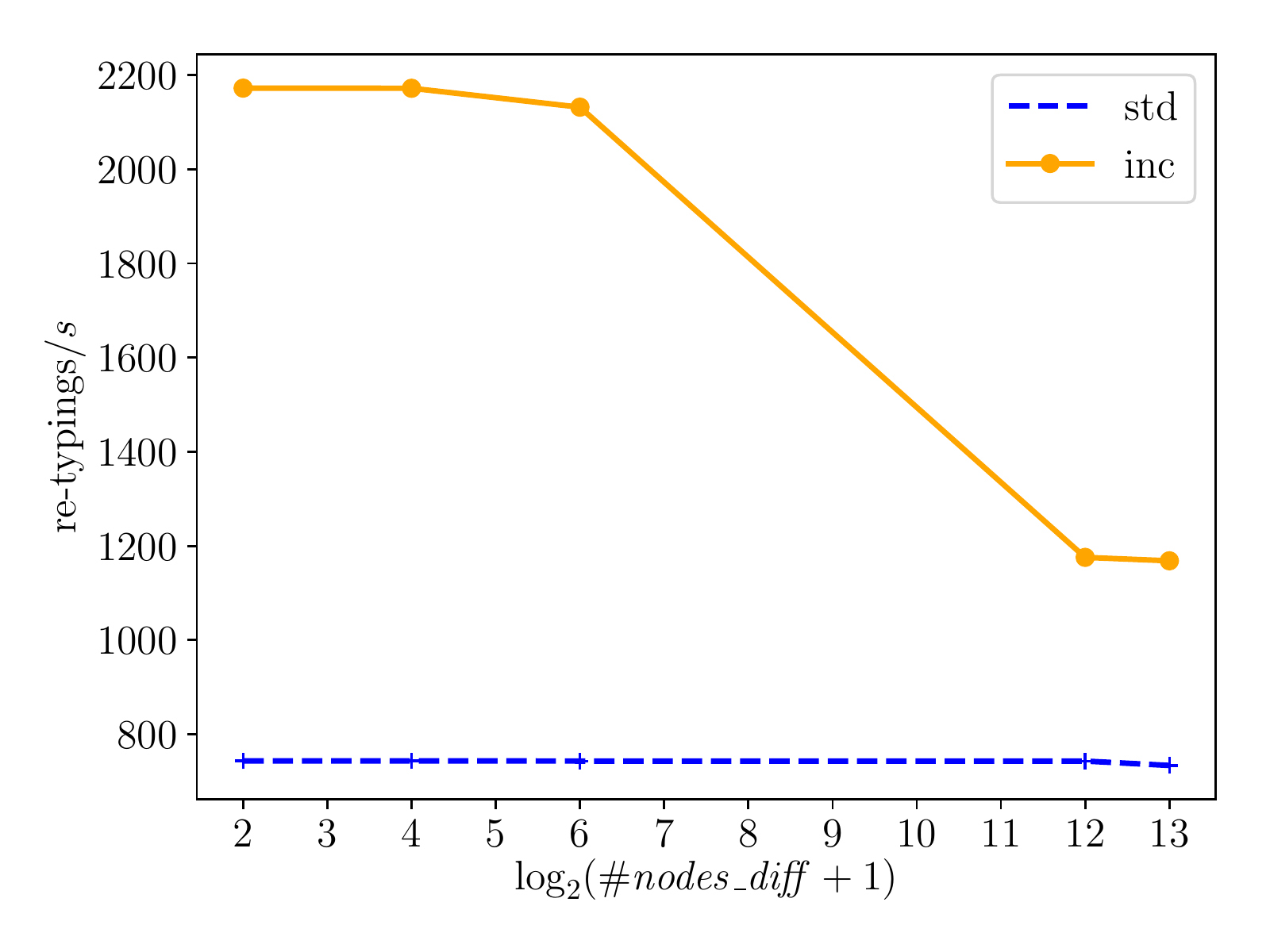}
    \caption{$\mathit{depth} = 14, \mathit{\#variables} = 2^{11}$}
  \end{subfigure}
    \begin{subfigure}[b]{0.47\textwidth}
    \includegraphics[width=\textwidth]{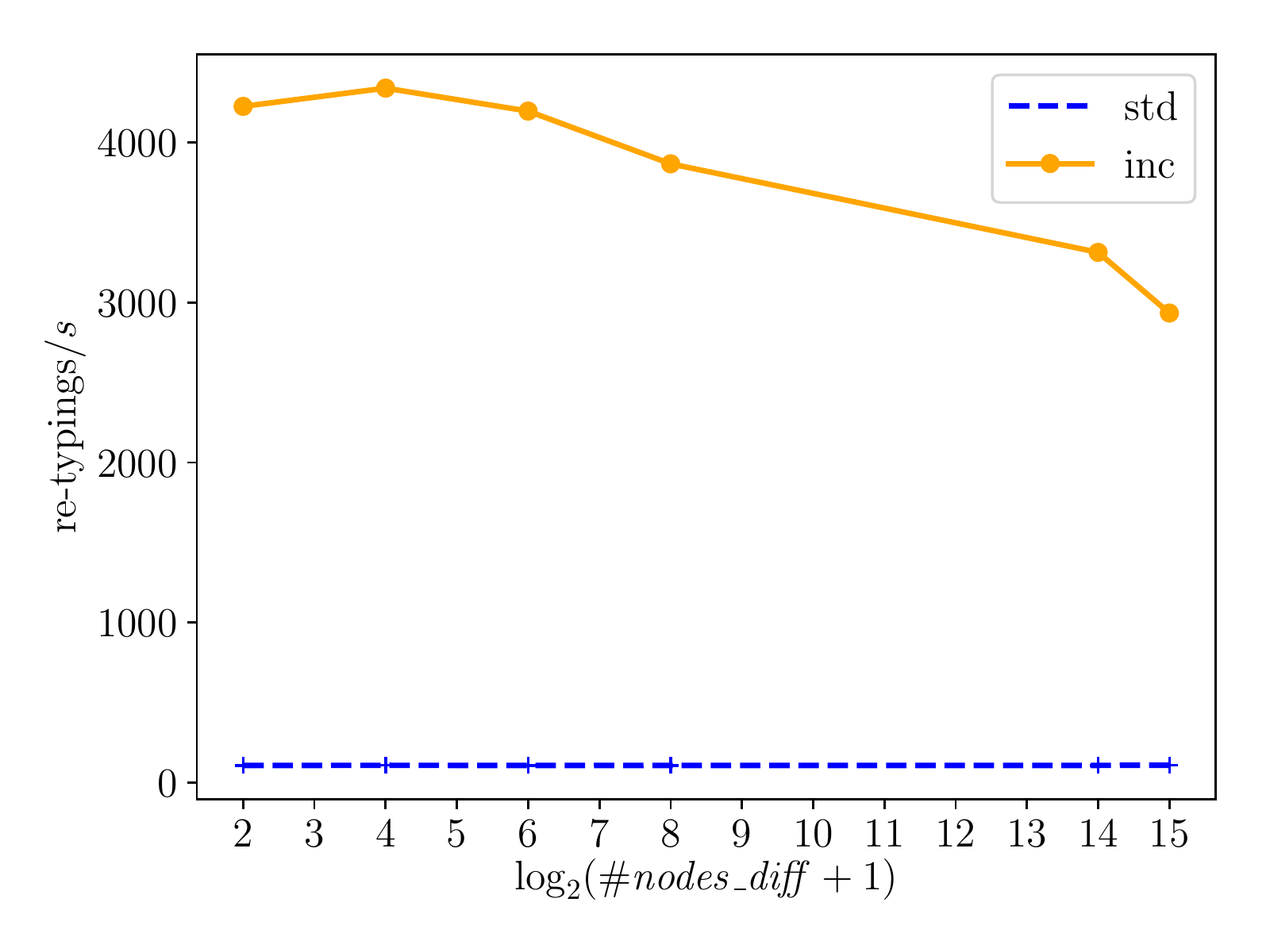}
    \caption{$\mathit{depth} = 16, \mathit{\#variables} = 2^{9}$}
  \end{subfigure}
  \hspace{1em}
  \begin{subfigure}[b]{0.47\textwidth}
    \includegraphics[width=\textwidth]{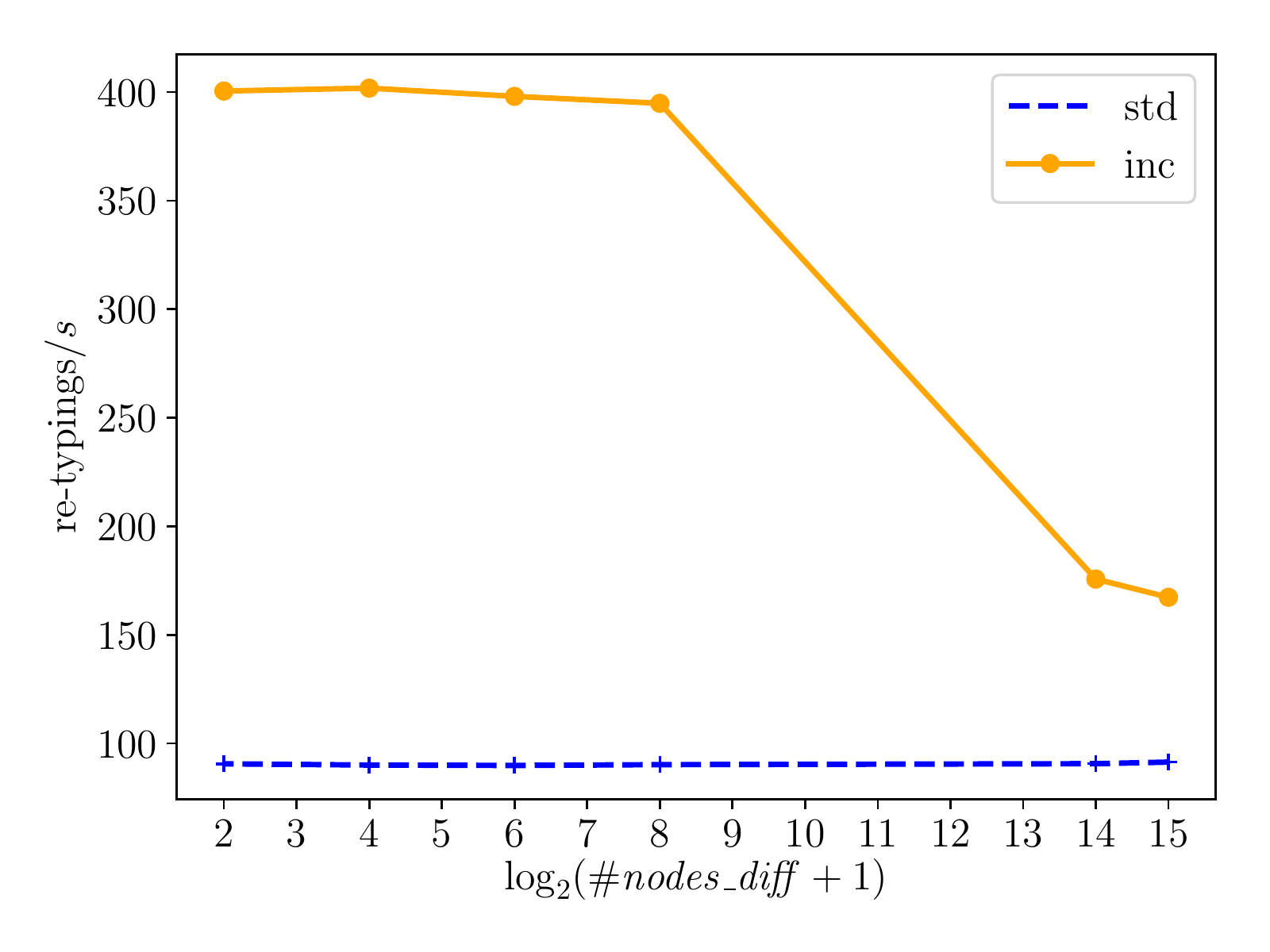}
    \caption{$\mathit{depth} = 16, \mathit{\#variables} = 2^{13}$}
  \end{subfigure}
  \caption{Some further experimental results comparing the number of re-typings per second vs. the number of nodes of \emph{diff}.
  The blue, dashed plot is for the standard type checking, while the orange, solid one is for the incremental usage.}\label{fig:more-diff-cmp}
\end{figure}

\end{document}